\documentclass[sigconf]{acmart}
\usepackage{algorithm2e}
\usepackage{caption}
\usepackage{booktabs}
\usepackage{subcaption}
\usepackage{multirow}
\usepackage{lipsum}
\usepackage{float}
\usepackage{pgfplots}
\usepackage{pgfplotstable}
\usepackage{pgf-pie}
\usepackage{setspace}
\usepgfplotslibrary{fillbetween}
\pgfplotsset{width=10cm,compat=1.9}
\usepackage{thmtools, thm-restate}
\usepackage{enumitem}

\usepackage{pifont}
\newcommand{\cmark}{\ding{51}}%
\newcommand{\xmark}{\ding{55}}%

\AtBeginDocument{%
  }

\copyrightyear{2025}
\acmYear{2025}
\setcopyright{acmlicensed}
\acmConference[WWW '25] {Proceedings of the ACM Web Conference 2025}{April 28--May 2, 2025}{Sydney, NSW, Australia.}
\acmBooktitle{Proceedings of the ACM Web Conference 2025 (WWW '25), April 28--May 2, 2025, Sydney, NSW, Australia}
\acmISBN{979-8-4007-1274-6/25/04}
\acmDOI{10.1145/3696410.3714921}

\begin{document}

\title{On-device Content-based Recommendation with Single-shot Embedding Pruning: A Cooperative Game Perspective}

\author{Hung Vinh Tran}
\email{h.v.tran@uq.edu.au}
\affiliation{%
  \institution{The University of Queensland}
  \city{Brisbane}
  \state{Queensland}
  \country{Australia}
  \postcode{4072}
}
\author{Tong Chen}
\email{tong.chen@uq.edu.au}
\affiliation{%
  \institution{The University of Queensland}
  \city{Brisbane}
  \state{Queensland}
  \country{Australia}
  \postcode{4072}
}

\author{Guanhua Ye}
\email{g.ye@bupt.edu.cn}
\affiliation{%
  \institution{Beijing University of Posts and Telecommunications}
  \city{Beijing}
  \country{China}
}

\author{Quoc Viet Hung Nguyen}
\email{henry.nguyen@griffith.edu.au}
\affiliation{%
  \institution{Griffith University}
  \city{Gold Coast}
  \state{Queensland}
  \country{Australia}
  \postcode{4222}
}

\author{Kai Zheng}
\email{zhengkai@uestc.edu.cn}
\affiliation{%
  \institution{University of Electronic Science and Technology of China}
  \city{Chengdu}
  \country{China}
}

\author{Hongzhi Yin}
\email{h.yin1@uq.edu.au}
\authornote{Corresponding author.}
\affiliation{%
  \institution{The University of Queensland}
  \city{Brisbane}
  \state{Queensland}
  \country{Australia}
  \postcode{4072}
}

\renewcommand{\shortauthors}{Tran et al.}

\begin{abstract}
Content-based Recommender Systems (CRSs) play a crucial role in shaping user experiences in e-commerce, online advertising, and personalized recommendations. However, due to the vast amount of categorical features, the embedding tables used in CRS models pose a significant storage bottleneck for real-world deployment, especially on resource-constrained devices. To address this problem, various embedding pruning methods have been proposed, but most existing ones require expensive retraining steps for each target parameter budget, leading to enormous computation costs. In reality, this computation cost is a major hurdle in real-world applications with diverse storage requirements, such as federated learning and streaming settings. In this paper, we propose  \textbf{Sha}pley \textbf{V}alue-guided \textbf{E}mbedding \textbf{R}eduction (Shaver) as our response. With Shaver, we view the problem from a cooperative game perspective, and quantify each embedding parameter's contribution with Shapley values to facilitate contribution-based parameter pruning. To address the inherently high computation costs of Shapley values, we propose an efficient and unbiased method to estimate Shapley values of a CRS's embedding parameters. Moreover, in the pruning stage, we put forward a field-aware codebook to mitigate the information loss in the traditional zero-out treatment. Through extensive experiments on three real-world datasets, Shaver has demonstrated competitive performance with lightweight recommendation models across various parameter budgets. The source code is available at \textit{https://github.com/chenxing1999/shaver}.
  
\end{abstract}

\maketitle

\newcommand{\meanE}[0]{\mathbb{E}}
\newcommand{\len}[1]{| #1 |}

\section{Introduction}
In the modern era, recommender systems (RSs) play a vital role in assisting users to identify relevant information. 
Meanwhile, on top of user-item interactions, the increasingly available side features of users and items (e.g., user gender and product category) have made content-based RSs (CRSs) \cite{marcuzzo2022recommendation} a step above collaborative filtering-based RSs, especially in cold-start settings. 
According to the Statista report \cite{statista2024}, in the US, online advertising -- a major CRS application -- has doubled its market size from \$107.5 billion in 2018 to \$225 billion in 2023. More recently, with privacy legislation and cyber security awareness on the rise, on-device deployment of CRSs is quickly gaining popularity, where examples include Brave Browser's federated news recommendation \cite{minto2021}, as well as Kuaishou's short video recommendation on mobile devices \cite{gong2022real}.

Most modern CRSs depend on sparse categorical features, such as movie genres and user regions in movie recommendation \cite{marcuzzo2022recommendation}, where a dense embedding table that hosts all possible features' embeddings (i.e., vector representations) is trained. 
Given the enormous amount of sparse categorical features, embedding tables have become the storage bottleneck for modern CRSs, which contradicts the surging demand for scalability and efficiency during on-device deployment \cite{yin2024device,memcom2020}. 
For illustration, in click-through rate (CTR) prediction, a typical CRS task that predicts users' clicking intention on an item, the RS deployed by Meta has been reported to consume 12T parameters and can demand up to 96TB of memory and multiple GPUs to train \cite{metarecsys2022}. 
Consequently, lowering the parameter footprint has attracted immense attention in both industry applications \cite{memcom2020,nis2019} and academic research \cite{bet2024,han2021deeprec}.
For instance, PEP \cite{pep2021} prunes the embedding table by searching a pruning mask through a learnable threshold. After the initial training of both the threshold and the model, the model is retrained from scratch to work under the pruned setting.
Another example is AutoSrh \cite{kong2022autosrh}, which adaptively groups features into blocks. After bi-level optimization, a sparse embedding parameterization can be selected according to the memory budget and then optimized via retraining. 

Generally, these methods focus on pruning the embedding table to a fixed size, hence an expensive training cycle is required for every single target sparsity. 
On one side, in most existing solutions, optimal trade-offs between recommendation accuracy and parameter efficiency require tailoring models according to each device's specific computational and storage capacities.
However, on the flip side of the coin, such expensive, repetitive training procedures are not affordable in practical settings with varying on-device storage needs. We provide two scenarios below:
\begin{itemize}[leftmargin=*]
    \item \textbf{Scenario 1:} Federated recommendation ensures privacy by assigning each user an on-device recommender. Such paradigms commonly suffer from hardware heterogeneity, where the variation of computing resources across devices (e.g., TV boxes, smart watches, and mobile phones) \cite{rule2021,yin2024device,yin2024device_tutorial} prohibits a uniform model architecture from being used. This naturally calls for more efficient means of embedding pruning to provide multiple versions of a trained model.
    \item \textbf{Scenario 2:} When deployed on-device, a model needs to operate under different computational resource budgets due to software parallelization and varying battery status \cite{yu2019slimmable,cai2020once}, and CRSs are no exception. Considering the time-sensitive nature of recommendation services, a sophisticated embedding pruning mechanism that requires reinforcement learning \cite{yunke_continous2023,bet2024} or even post-pruning retraining \cite{pep2021} will lead to a
    significant overhead, hindering the usability of the service pipeline. 
\end{itemize}
Ideally, considering those real-life scenarios, a competent pruning technique is expected to quickly adapt a well-trained CRS model to any specified parameter budget, without undergoing iterative search or retraining that are both time-consuming. 


Therefore, this paper focuses on a more practical, yet largely neglected setting of CRS compression for on-device deployment, which is to prune the embedding table of a recommender in one single shot\footnote{To clarify, different from transfer learning, our single-shot setting refers to performing model pruning in only one forward (and sometimes backward) pass \cite{lee2018snip,SSEDS}.}. 
To achieve single-shot pruning, a common practice is based on attribution scores \cite{taylor2017,ancona2020shapley}. In a nutshell, by quantifying the contribution from each individual model weight (i.e., embedding parameter in our case) to a performance metric, parameters with the lowest scores can be nullified to meet an arbitrary size constraint. In content-based recommendation, existing solutions \cite{SSEDS,wang2023single} perform parameter attribution by measuring the loss value change after removing the parameter from the full embedding table, which can be approximated with first-order gradients \cite{lee2018snip}. If we treat all embedding parameters as cooperative ``players'' in a game, then this is essentially asking the question: \textit{In a team formed by all players, what is each individual player's contribution to the whole team?} 
However, this formulation hardly aligns with the pruning context, where a player (i.e., embedding parameter), if not pruned, will actually form a smaller team with only a subset of players. In a single-shot pruning setting, instead of retaining embedding parameters that are most important to the full parameter set, we should ultimately identify parameters that can be of high value to any parameter subsets. So, the right question to be asked is: \textit{What is the expected contribution of each individual player, when it joins any team formed by a subset of players?} Due to the intertwined dependency among parameters, it is clear that this question does not share the same answer with the former one. 

To answer this question, we aim to investigate embedding pruning for CRSs via the lens of cooperative games, where we employ Shapley value \cite{roth1988shapley} to perform embedding parameter attribution. Shapley value is a well-defined method in the cooperative game theory for fairly allocating the total game revenue to players based on their individual contributions \cite{shap2017,sage2020,jethani2021fastshap}. It quantifies each player's contribution by considering all possible team combinations, resulting in the expected contribution of the player in any group. 
As such, compared with existing practices \cite{lee2018snip,SSEDS}, Shapley value is a better fit for generating unbiased attribution scores of embedding parameters to facilitate fast recommendation model pruning. 


However, some major challenges are yet to be resolved before we can enjoy the benefit of Shapley values. Firstly, the computational costs of Shapley value are prohibitively expensive. In our context, when a CRS only activates a subset of embedding parameters, each parameter's contribution in this subset can be measured by comparing the CRS's recommendation performance with and without it. Then, a parameter's Shapley value is its average contribution to all possible subsets of parameters. To calculate all the exact Shapley values, we need to enumerate over all $2^{nd}$ parameter subsets, where $n$ and $d$ are respectively the numbers of features and embedding dimensions. Considering that in CRSs $nd$ is intractable (e.g., reaching 100 million in our biggest dataset), and reliably measuring the recommendation performance requires a sufficiently large dataset, the number of resulted forward passes has rendered the straightforward calculation of Shapley values an infeasible option. 
Secondly, most pruning methods employ a simple zero-out approach to remove less important parameters and create a sparse embedding matrix.
While this reduces the model size, it can result in significant information loss as recommenders rely heavily on dense feature representations for accurate predictions.
This further deteriorates in CRSs due to their extensive use of dot product \cite{deepfm2017,fm2010} and element-wise multiplication \cite{dcnv2} between embeddings to model feature interactions. By setting parameters to zero, pruning can disrupt the complex relationships between features, leading to suboptimal performance. 

Motivated by these challenges, we propose \textbf{Sha}pley \textbf{V}alue-guided \textbf{E}mbedding \textbf{R}eduction (\textbf{Shaver}), a novel embedding pruning method for compressing a CRS model to any target size in one shot. Shaver assigns each embedding parameter a Shapley value as its attribution score, such that a well-trained model can prune its embeddings to any parameter budget in a single shot. To address the high computational costs of obtaining exact Shapley values, through a series of theoretical analyses, we develop an unbiased approximation method to reliably estimate the Shapley values of all embedding parameters. 
Efficiency-wise, Shaver reduces the computational complexity from $O((nd)! \times \len{\mathcal{D}})$ to the much lower $O(md \times \len{\mathcal{D}})$, with $\mathcal{D}$ and $m$ respectively being the dataset and number of feature fields (e.g., user gender and movie genre). Notably, $m\ll n$ in our case -- tens of feature fields versus millions of features (see Table \ref{tab:dataset-stat}).
To reduce the information loss caused by the default zero-padding on pruned parameters, we propose an alternative that searches for an optimal placeholder value per feature field, termed the codebook, which is then used for padding nullified embedding parameters during inference. In addition, the codebook values can be solved in a closed form, making it better at retaining the recommendation performance with a negligible computational overhead. 

In summary, our main contributions are:
\begin{itemize}
    \item We revisit one important but largely neglected task of single-shot embedding pruning, which is more relevant in the real-world, on-device deployment of CRSs. 
    \item We propose Shaver, a novel embedding compression method that reduces the time complexity of Shapley value estimation for CTR models and provides a more effective alternative to the widely used zero-out approach in model pruning with a field-aware codebook.
    \item We conduct extensive experiments on three datasets to compare Shaver with state-of-the-art embedding compression methods. The results show that even without an extra fine-tuning step, Shaver can achieve competitive performance.
\end{itemize}

\section{Related Work}
\subsection{CTR Prediction}

CTR prediction plays a vital role in the development of online advertisements. 
One of the first and most influential backbones for CTR prediction is Factorization Machine (FM) \cite{fm2010}, where the interaction between features are explicitly model upto 2nd order.
Recently, with the rapid development of deep neural networks (DNN), various deep models have also been proposed.
NeuMF \cite{neumf2017} combines the introduced Generalized Matrix Factorization (GMF) with a DNN network branch.
DeepFM \cite{deepfm2017} is another prime example, which adds a DNN branch into the FM model. 
Unlike DeepFM, which only attempts to model up to 2nd-order interaction, DCNv2 \cite{dcnv2} explicitly models the $l + 1$-th order interaction by $l$ crossnet layers. DCNv2 also incorporates a DNN to model more complex interactions. 
To further reduce the computation cost of DCNv2, Wang et al. \cite{dcnv2} also propose integrating Mixture-of-Experts for modeling the cross-layers, creating DCN--Mix.

\subsection{Lightweight Recommender Systems}
A diverse set of research endeavors has proposed various methods for learning light-weight embeddings \cite{zhang2023experimental, tran2024thorough, li_survey2024}. Compositional embedding methods decompose the original embedding table into multiple smaller meta-embedding tables and share parameters between features \cite{qr2019,ttrec2021,dhe2021,memcom2020,legcf2024,coleman2023unified,cel2023,xia2023towards}. 
Quantization approaches are also applied in light-weight RSs \cite{alpt2023,xu2021agile,yang2020mixed,guan2019post}, which focus on reducing the loss of model quality in lower precision representation. 
However, these approaches generally suffer from limited compression ratios \cite{zhang2023experimental}.
Knowledge distillation \cite{xia2022device,xia2023efficient,wang2020next} is employed to mitigate performance degradation in smaller models.
Another line of research focuses on finding the best parameter allocation through Neural Architecture Search (NAS), which usually involves reinforcement learning \cite{nis2019, yunke_continous2023,qu2024scalable}, bi-level optimization \cite{autoemb2021,autodim2021,kong2022autosrh,zheng2024personalized}, or evolutionary search \cite{rule2021,optembed2022}. Consequently, NAS-based approaches generally have high computation costs. 
Pruning is another popular approach to reduce the amount of embedding parameters by finding a binary mask \cite{pep2021, SSEDS,wang2024dynamic,amtl2021,qu2024sparser}. 
Recently, many hybrid methods \cite{cerp2023,optembed2022} have also been proposed to combine the benefit of various archetypes to achieve the best trade-off between performance and parameter budget.

Our method falls within the pruning category. 
Unlike most of other pruning methods \cite{pep2021,amtl2021}, we eliminate the hassle of simultaneously optimizing both the pruning budget mask and model performance. 
Compared to \cite{qu2024scalable,rule2021}, which also can adapt to arbitrary parameter budgets but focus on collaborative filtering, our method tailored towards CRSs. 
Additionally, unlike SSEDS \cite{SSEDS}, which mandates uniform dimensions across all features within the same field, our method offers greater flexibility by allowing varied budget allocations without such constraints.
Moreover,  our single-shot pruning setting focuses on efficiently customizing a trained CRS to heterogeneous memory budgets, which is different from online pruning settings \cite{lai2023adaembed,liu2025cafe,zhang2024cafe} where data arrives continually.

\subsection{Shapley Value for Machine Learning}
Shapley values have been widely adopted in the explainable AI community due to its high correlation with human intuition \cite{shap2017}. Because of its high computation cost, various researches have been proposed to tackle this challenge through 
random order sampling \cite{montecarloShapley,sage2020}, 
solving weighted least square problem \cite{shap2017,covert2020improving},
being model specific \cite{shap2017,treeshap2020},
and recently amortization approaches \cite{jethani2021fastshap,covert2024stochastic}. Recently, the success of Shapley values in explainable AI has inspired its application in pruning deep models. 
To start with, Ancona et al. \cite{ancona2020shapley} pointed out several desirable properties of Shapley values in pruning neurons:
\begin{itemize}
    \item Null player: If the game (loss) doesn't depend on a player (parameter or group of parameters), its attribution is zero.
    \item Symmetry: If the game depends on two players equally, then the two players should receive the same attribution.
    \item Efficiency: Attributions sum up to the difference between the loss evaluated when all players are enabled and the loss evaluated when all players are removed.
\end{itemize}
Later, NeuronShapley \cite{ghorbani2020neuron} considered Shapley value estimation as a multi-arm bandit problem, and applied truncated Monte Carlo to estimate Shapley values. 
Leveraging the multi-arm bandit settings, Guan et al. \cite{guan2022few} employ an $\epsilon$-greedy approach.
Troupe \cite{troupe2021} proposes using Shapley values to select models for an ensemble based on a specific data point.
Last, although it is not applied for compressing embedding tables, Shapley values also have been applied in recommender system designs \cite{ben2018game,fang2023autoshape}.

\section{Preliminaries}
In this section, we first introduce the key concepts of content-based recommendation by defining a classic task, namely click-through rate (CTR) prediction. Then, we provide the definition of Shapley value, the theoretical foundation of our solution.

\subsection{Content-based Recommendation}
In CTR prediction, a dataset $\mathcal{D}$ is a set of samples $(\mathbf{x}, y)$, where $\mathbf{x}$ is a multi-dimensional vector combining both user and item features, and $y$ is a binary label indicating positive ($y=1$, clicked) or negative ($y=0$, not clicked) user interaction with the item. For $\mathbf{x}$, its features are commonly a collection of sparse binary features from multiple fields (i.e., user occupation and product category) \cite{marcuzzo2022recommendation,bars2022}, such as:
\begin{equation}
    \mathbf{x} = 
    \underbrace{[1, 0]}_{\mathbf{x}_1( \text{field 1})} \underbrace{[0,\dots,1,\dots,0]}_{\mathbf{x}_2(\text{field 2})}\,\,{\dots}\,\, \underbrace{[0,\dots,1,\dots,0]}_{\mathbf{x}_m(\text{field}\;m)},
\end{equation}
where an example of $\mathbf{x}_1$ and $\mathbf{x}_2$ can be $user\, gender=female$ in field 1 and $movie\, genre=action$ in field 2. 
$\mathbf{x}$ is then a concatenation of all categorical feature encodings $\{\mathbf{x}_j\}_{j=1}^m$ from all $m$ fields, where numerical features are commonly discretized via bucketing \cite{bars2022} in CTR prediction tasks. 

In real-world applications, the input vector $\mathbf{x}$ is extremely sparse, therefore each field $j$'s feature is mapped into a $d$-dimensional, dense latent vector $\mathbf{e}_j \in \mathbb{R}^d$, called embedding:
\begin{equation}
    \mathbf{e}_j = \mathbf{V}_j^{\top}{\mathbf{x}_j},
\end{equation}
where $\mathbf{V}_j \in \mathbb{R}^{n_j \times d}$ and $\mathbf{x}_j\in \mathbb{R}^{n_j}$ are respectively the embedding table and sparse encoding for field $j$, with $n_j$ the number of features in field $j$. In practice, all $m$ field-specific embedding tables $\mathbf{V}_i$ are concatenated as a single embedding table $\mathbf{E} \in R^{n \times d}$ for storage, where $\mathbf{E}= \left[ \mathbf{V}_1, \mathbf{V}_2, \dots, \mathbf{V}_m \right]$, and $n = \sum_{j=1}^m n_j$ is the total number of features. 
After obtaining embeddings for $\mathbf{x}$ from all fields $\{ \mathbf{e}_1, \mathbf{e}_2, \dots, \mathbf{e}_m \}$, they are further processed by a recommendation backbone model:
\begin{equation}
    \hat{y} = f(\{ \mathbf{e}_1, \mathbf{e}_2, \dots, \mathbf{e}_m \}),
\end{equation}
with $\hat{y}$ is a binary classification prediction whether a user will click on an item or not.
Its primary target is to model the complex interactions between different features. For instance, the model can combine the user gender and the movie genre to create a 2nd-order interaction, enhancing model performance \cite{fm2010, deepfm2017}.
These high order interactions are commonly modeled explicitly by dot product \cite{fm2010, deepfm2017} or element-wise multiplications \cite{dcnv2}, and implicitly by deep neural networks \cite{dcnv2,deepfm2017}. 
Finally, $f(\cdot)$ is trained with log loss:
\begin{equation}
    L = -y \ln (\hat{y}) - (1 - y) \ln (1 - \hat{y}),
\end{equation}
which quantifies the prediction error between $y$ and $\hat{y}$.

\begin{figure*}[t!]
    \centering
    \includegraphics[width=0.7\linewidth]{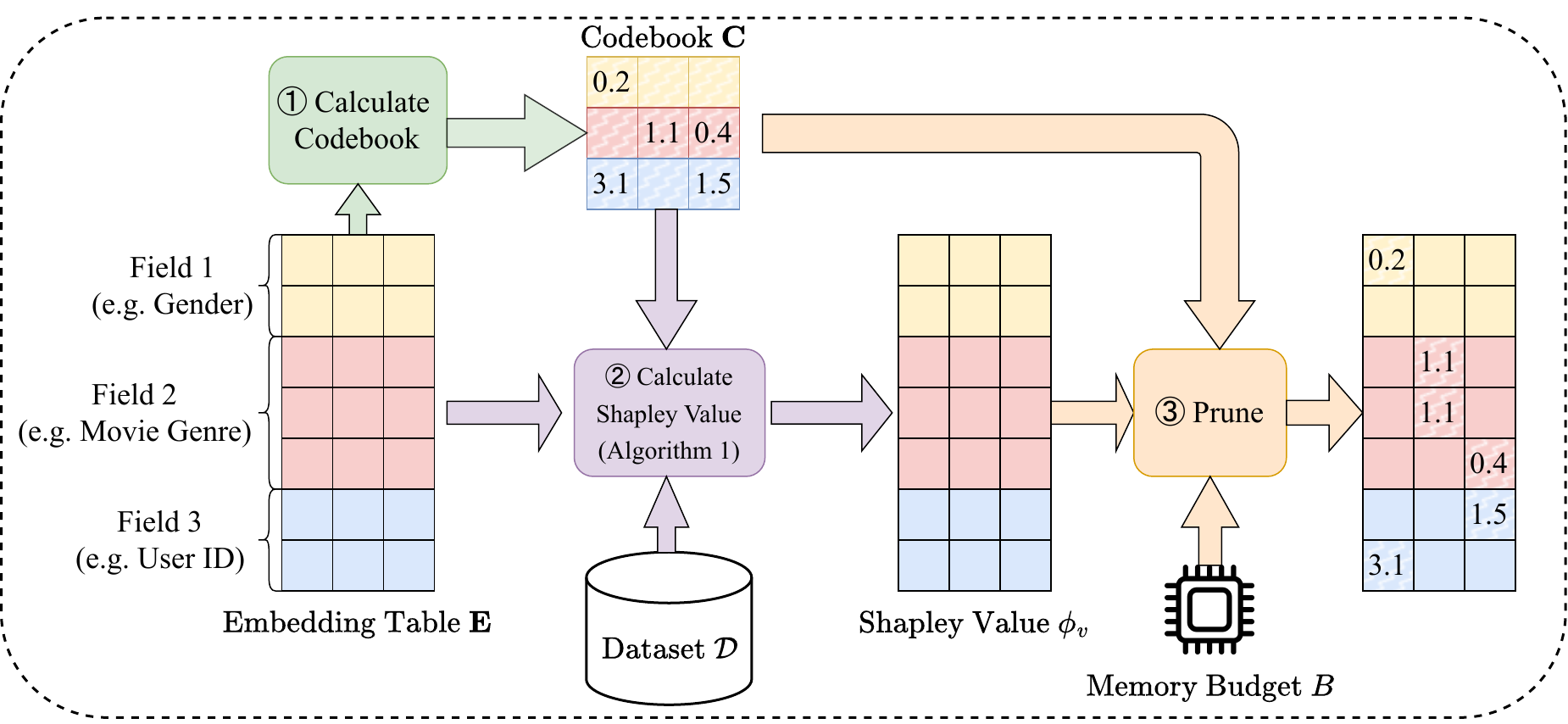}
    \vspace{-.3cm}
    \caption{The overview of Shaver. 
    \textmd{We calculate placeholder values (codebook $\mathbf{C}$), and then compute Shapley values $\phi_v$ from the provided dataset $\mathcal{D}$. On any required memory budget $B$, we replace embedding parameters with the lowest attribution scores by placeholder values.}
    }
    \label{fig:overview}
    \vspace{-.3cm}
\end{figure*}

\subsection{Shapley Value}\label{sec:SV_def}
Shapley value originates from the cooperative game theory, whose aim is to assign each player a score that quantifies the player's contribution in a cooperative game. Formally, a cooperative game consists of a set of players denoted by $\mathcal{N}$, where any arbitrary subset of players can form a ``team''. The coalitional values generated by player subsets are measured by function $v: 2^{\mathcal{N}} \rightarrow \mathbb{R}$, with $v(\emptyset)=0$. 
For player $i\in\mathcal{N}$, its Shapley value ${\phi}_v(i)$ w.r.t. function $v(\cdot)$ is:
\begin{equation}
    {\phi}_v(i) = \frac{1}{|\mathcal{N}|} \sum_{ \mathcal{S} \subseteq \mathcal{N} \setminus i } \binom{|\mathcal{N}| -1}{|\mathcal{S}|}^{-1} \Big{(}v(\mathcal{S} \cup i) - v(\mathcal{S})\Big{)},
\end{equation}
where $\binom{\cdot}{\cdot}$ is the binomial coefficient, and $v(\mathcal{S} \cup i) - v(\mathcal{S})$ quantifies the marginal contribution of $i$ after joining team $\mathcal{S}$. In a nutshell, ${\phi}_v(i)$ is computed by enumerating over all possible teams $\mathcal{S}\subseteq \mathcal{N} \setminus i$. 
With that, a commonly used, permutation-based formulation of the Shapley value is:
\begin{equation}
    \phi_v(i) = \frac{1}{\len{\mathcal{N}}!} \sum_{\forall \mathcal{R}} \left( v ( \mathcal{P}^{\mathcal{R}}_i \cup i ) - v ( \mathcal{P}^{\mathcal{R}}_i ) \right),
    \label{eq:shapley-monte}
\end{equation}
where the ordered set $\mathcal{R}$ is a possible permutation of all players, $\mathcal{P}^{\mathcal{R}}_i$ is the set of players that precede $i$ in $\mathcal{R}$. Thus, the Shapley value $\phi_v(i)$ can be calculated as the mean/expectation of $i$'s contributions in all $\len{\mathcal{N}}!$ possible permutations of $\mathcal{N}$.

\textbf{Monte Carlo Approximation of Shapley Values.} Based on Eq.(\ref{eq:shapley-monte}), the default form of Shapley value calculation incurs a factorial time complexity w.r.t. the number of players. To scale it to a large $\mathcal{N}$, approximating the Shapley value via Mote Carlo approach is a widely adopted workaround \cite{ghorbani2020neuron,ancona2020shapley}, which estimates 
$\phi_v(i)$ by 
sampling $p$ permutations $\mathcal{R}_k$ and take the average marginal contribution in these samples to approximate Shapley value:
\begin{equation}
    {\phi}_v(i) \approx \frac{1}{p} \sum_{k=1}^{p} \left( v ( \mathcal{P}^{\mathcal{R}_k}_i \cup  i ) - v ( \mathcal{P}^{\mathcal{R}_k}_i ) \right),
\end{equation}
where ${\phi}_v(i)$ is now the Monte Carlo approximation \cite{montecarloShapley} of  player $i$'s Shapley value w.r.t. value function $v(\cdot)$.

\section{Shaver: The Proposed Method}
\newcommand{\featind}[0]{i}
\newcommand{\fieldind}[0]{j}
\newcommand{\columnind}[0]{c}

In this section, we unfold the design of Shaver, namely \underline{Sha}pley \underline{V}alue-guided \underline{E}mbedding \underline{R}eduction for on-device CRSs, of which an overview is provided in Figure \ref{fig:overview}. We firstly formally define the problem, and subsequently provide corresponding theoretical analyses and pseudocode to efficiently estimate the Shapley value of each embedding parameter. 
Finally, we derive a codebook to replace the traditional zero-padding strategy, so as to compensate for the information loss brought by embedding parameter pruning and maintain maximum model performance. 




\subsection{Shapley Values of Embedding Parameters}

Before defining the problem, we provide some clarifications for Shapley value in our specific task context. To facilitate selective embedding table pruning, we derive a Shapley value for every single parameter in the feature embedding table $\mathbf{E}$. Let $\mathbf{E}[\featind,\columnind]$ denote the parameter at the $\featind$-th row and $\columnind$-th column (i.e., the $\columnind$-th embedding dimension of feature $\featind$), the Shapley value of player $(\featind,\columnind)$ is denoted as $\phi_v(\featind,\columnind)$ with a value function $v(\cdot)$. Please note that the use case of Shapley value in our work differs from the mainstream, data-centric counterparts \cite{shap2017,jethani2021fastshap} that focus on calculating values of individual data points or raw features, while we are specifically interested in the Shapley values of model parameters.



\textbf{Value Function $v(\cdot)$.} 
For measuring $\phi_v(\featind,\columnind)$, we need to define a value function $v(\cdot)$ that can quantify the contribution from a group of players (i.e., parameters) $\mathcal{S}=\{\mathbf{E}[\featind_1,\columnind_1], \mathbf{E}[\featind_2,\columnind_2],\dots,\mathbf{E}[\featind_s,\columnind_s]\}$. 
Given a CTR model $f(\cdot)$, we define $v(\cdot)$ as the log loss increase when parameters $\mathcal{S}$ become absent from $\mathbf{E}$:
\begin{equation}\label{eq:value_f}
    v(\mathcal{S}) = \frac{1} { \len{\mathcal{D}} } \sum_{(\mathbf{x}, y )\in \mathcal{D} }  \Big{(}  L(f_{\text{w/o-}\mathcal{S}}(\mathbf{x}), y) - L(f(\mathbf{x}), y) \Big{)},
\end{equation}
where $f_{\text{w/o-}\mathcal{S}}(\cdot)$ denotes the CTR model without embedding parameters $\mathcal{S}$, $L$ is the log loss, and $\mathcal{D}$ denotes the dataset.  
With this definition, the Shapley value $\phi_v(\featind,\columnind)$ of a single parameter $\mathbf{E}[\featind,\columnind]$ can be quantified via an analogous process as in Section \ref{sec:SV_def}.

\begin{restatable}[Single-shot Pruning with Shapley Values]{definition}{localeq}
    \label{definition:task}
    Given a trained CTR model $f(\cdot)$, our goal is to obtain the Shapley value of every single parameter $\phi_v(\featind,\columnind)$ in the feature embedding table. Based on $\phi_v(\featind,\columnind)$, we can efficiently prune a CTR model for any on-device budget by directly removing the lowest ranked entries from its embedding table to meet the specified parameter size.
\end{restatable}



\subsection{Local and Global Values}  
If straightforward Monte Carlo approximation of ${\phi}_{v}(\featind,\columnind)$ is in use, for all $\len{\mathcal{N}} = nd$ embedding parameters, $\len{\mathcal{N}} \len{\mathcal{D}}= n d  \len{\mathcal{D}}$ forward passes are needed per permutation, leading to a prohibitive computational cost of $\mathcal{O}(pnd \len{\mathcal{D}})$. 
One main cause is that Eq.(\ref{eq:value_f}) lays emphasis on embedding parameters' values to the overall data distribution. We term this \textit{global values}, whose calculation involves a sufficiently large, if not the full dataset $\mathcal{D}$. With that said, the function of \textit{local values} $v^{loc}(\cdot)$ is defined based on single data samples:
\begin{equation}
    v^{loc}(\mathcal{S}|\mathbf{x},y) =  L(f_{\text{w/o-}\mathcal{S}}(\mathbf{x}), y) - L(f(\mathbf{x}), y).
\end{equation}
In contrast to computing the global values, there are some efforts that draw an individual data sample at a time, pair it with one random permutation, and then derive each parameter's local value to approximate the global one~\cite{sage2020}. This process can be described as follows: 
\begin{equation}
\begin{split}
 \phi_v(i,c) &= \meanE_{\mathcal{R}} \left[v (\mathcal{P}^\mathcal{R}_{i,c} \cup  (i,c)) - v ( \mathcal{P}^\mathcal{R}_{i,c} ) \right] 
        \\
        &= \meanE_{\mathcal{R}} \meanE_{(\mathbf{x}, y)} \left[ v^{loc}(\mathcal{P}^\mathcal{R}_{i,c} \cup  (i,c)|\mathbf{x},y) - v^{loc}(\mathcal{P}^\mathcal{R}_{i,c}|\mathbf{x},y) \right] \\
        &= \meanE_{(\mathbf{x}, y),\mathcal{R}} \left[ L(f_{\text{w/o-}\mathcal{S}_1}(\mathbf{x}), y) - L(f_{\text{w/o-}\mathcal{S}_2}(\mathbf{x}), y) \right]\\
        &= \frac{1}{|\mathcal{D}|}\!\! \sum_{(\mathbf{x}, y)\in \mathcal{D},\mathcal{R}_{(\mathbf{x}, y)}}\!\!\! \Big( L(f_{\text{w/o-}\mathcal{S}_1}(\mathbf{x}), y) - L(f_{\text{w/o-}\mathcal{S}_2}(\mathbf{x}), y) \Big),\\
    \end{split}
    \label{eq:sage-eq}
\end{equation}
with $\mathcal{S}_1$ and $\mathcal{S}_2$ defined as: 
\begin{equation}
\begin{split}
    \mathcal{S}_1 &=\{\mathbf{E}[i',c']\,|\,(i',c')\in \mathcal{P}^\mathcal{R}_{i,c} \cup  (i,c)\},\\ 
        \mathcal{S}_2 &=\{\mathbf{E}[i',c']\,|\,(i',c')\in \mathcal{P}^\mathcal{R}_{i,c}\},\\
\end{split}
\end{equation}
where $\mathcal{R}_{(\mathbf{x}, y)}$ denotes one random permutation paired with data sample $(\mathbf{x}, y)$. , Eq.(\ref{eq:sage-eq}) shows that, with a quality, decent-sized $\mathcal{D}$, the global Shapley value can be confidently approximated as the mean of local values \cite{sage2020}. 
For all parameters, this method runs in $\mathcal{O}(nd|\mathcal{D}|)$. Despite being $p$ times faster than vanilla Monte Carlo and a subsampled $\mathcal{D}$ from the full training dataset can help further lower the complexity, the large feature size $n$ (normally in million-level) remains the efficiency bottleneck for parameter attribution. 
The following section shows how we further reduce those $nd|\mathcal{D}|$ forward passes -- which is essentially iterating over the dataset for $nd$ times -- to a much smaller number..
\subsection{Efficient Shapley Value Computation for CRS Embeddings}

Recall that in CRSs, features are categorized into $m$ fields. On top of the above defined local value $v^{loc}(\cdot)$, we further define a local value function $u^{loc}(\cdot)$ specific to feature fields:
\begin{equation}
    u^{loc}(\mathcal{S}'|\mathbf{x},y) = L \left( f_{\text{w/o-}\mathcal{S}'}(\mathbf{x}), y \right) - L\left( f(\mathbf{x}), y \right).
\end{equation}
Distinct from $v^{loc}(\cdot)$, in $u^{loc}(\cdot)$, a player $(\fieldind,c)$ is \textit{a group of parameters} that are located at the same column $c$ in field $\fieldind$'s embedding table $\mathbf{V}_j$.
To avoid cluttered notations, in this section, if not specified, both $v(\cdot)$ and $u(\cdot)$ are the \textit{local value functions} calculated with a single instance $(\mathbf{x},y)$.
Then, we define our notion of corresponding player as follows.
\begin{restatable}[Corresponding Player]{definition}{localeq}
    \label{definition:corresponding}
    Given player $(\featind,\columnind)$ in $v(\cdot)$ and player $(j,c')$ in $u(\cdot)$, and a data instance $(\mathbf{x},y)$, we call $(\featind,\columnind)$ the \textit{corresponding player} of $(\fieldind,c')$ if $c=c'$ and the $i$-th feature is activated, i.e., $\mathbf{x}[i]=1$.
\end{restatable}
With the above definition, we provide a key theorem for efficient Shapely value computation below:
\begin{restatable}[]{theorem}{localeq}
    \label{theorem:local-eq}
    Assume player $(i, c)$ in $v(\cdot)$ is corresponding with player $(j, c)$ in $u(\cdot)$ w.r.t. data instance $(\mathbf{x},y)$. Then, denote their Shapley values based on local value functions respectively as $\phi_v^{\mathbf{x},y}(i,c)$ and $\phi_u^{\mathbf{x},y}(j,c)$, we have $\phi_v^{\mathbf{x},y}(i,c) = \phi_u^{\mathbf{x},y}(j,c)$.
\end{restatable}

Our proof is provided in Appendix \ref{sec:proofs}. In short, we observe that $\phi_v^{\mathbf{x},y}(i, c)$ is equivalent to $\phi_u^{\mathbf{x},y}(j,c)$, as long as $(i,c)$ is the corresponding player of $(j,c)$. 
As computing $\phi_u^{\mathbf{x},y}(j,c)$ will suffice, only $md$ rather than $nd$ iterations over $\mathcal{D}$ are now required to compute $\phi_u^{\mathbf{x},y}(j,c)$ for all players in $u(\cdot)$. Taking our largest evaluation dataset KDD as an example (see Table \ref{tab:dataset-stat}), with $11$ ($m$) feature fields, over 6 million ($n$) features, and an embedding dimension of $d=16$, the $\mathcal{O}(md|\mathcal{D}|)$ complexity of our Shapley value computation over $\mathcal{O}(nd|\mathcal{D}|)$ given that $m\ll n$. The algorithm description and convergence analysis are both provided in Appendix \ref{sec:shaver_details}. 

\subsection{Replacing Zero-padding with Field-aware Codebook}
\label{sec:codebook}
In various machine learning fields like explainable AI, the choice of placeholder value (i.e. the value representing the missing value) is crucial, and can greatly affect the final performance \cite{chen2023algorithms_shapley}. 
However, most if not all CRS pruning algorithms opt for a simple zero-padding strategy. 
This approach often leads to suboptimal recommendation performance, as interactions between feature embeddings are modeled using dot product and element-wise operations. Setting certain parameters to zero eliminates such interactions with all other features, compromising system effectiveness. 
Therefore, in this work, we introduce a  approach to represent pruned features by a field-aware codebook $\mathbf{C} \in \mathbb{R}^{m \times d}$. As it only comes with $md$ elements, its storage overhead is negligible. 
For a parameter budget, let $\mathbf{E}_{\mathcal{Q}}$ denote the sparsified embedding table that nullifies the least important parameters identified in set $\mathcal{Q}$, the following equation specifies how we impute the pruned embedding table into $\mathbf{E}_{\mathcal{Q},\mathbf{C}}$ with the codebook:
\begin{equation}
  \mathbf{E}_{\mathcal{Q},\mathbf{C}} \left[ {i,c} \right] = \begin{cases}
        \mathbf{E}_{\mathcal{Q}}  \left[ {i,c} \right], \text{ if } \{i, c\} \notin \mathcal{Q} \\
        \mathbf{C} [ {j,c}], \text{ else}
    \end{cases},
\end{equation}
with $j$ the corresponding field of feature $i$. 
After calculating the attribution score, we will replace pruned parameters based on the above rule instead of only zeroing out.

To better preserve the model performance, we will choose codebook $\mathbf{C^{*}}$ that minimizes the expected Euclidean distance between the original embedding table and its imputed version:
\begin{equation}
    \mathbf{C}^{*} = \min_{\mathbf{C}} \meanE_{\mathbf{x} \sim \mathcal{D}, \len{\mathcal{{Q}}}=B} \left[ \lVert \mathbf{E}^{\top}\mathbf{x} - \mathbf{E}_{\mathcal{Q},\mathbf{C}}^{\top}\mathbf{x} \rVert_2^2 \right],
\end{equation}
where $B\in (0,nd)$ represents any possible number of pruned parameters, and $\mathcal{Q}$ is uniformly sampled such that $\len{{\mathcal{Q}}} = B$. With this formulation, the optimal codebook $\mathbf{C}^{*} \in \mathbb{R}^{m \times d}$ can be solved via the weighted average of all feature in each field:
\begin{equation}
    \mathbf{C}^{*}[j,:] = \frac{\sum_{i \in \mathcal{F}_j} p_i \mathbf{V}_j[i,:]}{\sum_{i \in \mathcal{F}_j} p_i},
\end{equation}
where $\mathcal{F}_j$ is the set of features in field $j$, and $p_i$ is the occurrence frequency of feature $i$.
The detailed deriviation is provided in Appendix \ref{sec:proofs}. This closed-form solution allows for efficient computation of codebook $\mathbf{C}^{*}$ to compensate for the information loss from the pruning stage. When deployed on-device, both the pruned sparse embedding table $\mathbf{E}_{\mathcal{Q}}$ and the codebook $\mathbf{C}^{*}$ are stored. 
When inference is needed, only the $m$ embeddings $\{\mathbf{e}_1, \mathbf{e}_2, ..., \mathbf{e}_m\}$ of the feature $\mathbf{x}$ will be imputed as the input to $f(\cdot)$ for on-device CTR prediction. 

\section{Experiments}
\renewcommand{\arraystretch}{0.8}

In this section, we conduct experiments to study the effectiveness of Shaver. Specifically, we are interested in answering the following research questions (RQs):
\begin{enumerate}
    \item [\textbf{RQ1}:] Compared with non-single-shot baselines, how does Shaver perform when compressing CRS models?
    \item [\textbf{RQ2}:] Compared with other single-shot baselines, how does Shaver perform under different parameter budgets? 
    \item [\textbf{RQ3}:] Most on-device CRS compression methods require model retraining/finetuning post pruning. Does Shaver also benefit from such a practice?
    \item [\textbf{RQ4}:] 
    What is the impact of dataset size to the recommendation performance of the pruned models? 
    \item [\textbf{RQ5}:] Does Shaver exhibit any patterns when pruning parameters? 
\end{enumerate}

\begin{table}[t!]
    \centering
    \caption{Statistics of the preprocessed datasets.}
    \vspace{-0.3cm}
    
    \begin{tabular}{lrrr}
        \toprule
        Name & \#Instances & \#Features & \#Fields \\
        \midrule
        Criteo & 45,840,617 & 1,086,810 & 39  \\
        Avazu & 40,428,967 & 4,428,511 & 22 \\
        KDD & 149,639,105 & 6,019,761 & 11 \\
        \bottomrule
    \end{tabular}
    \vspace{-0.3cm}
    \label{tab:dataset-stat}
\end{table}
\renewcommand{\arraystretch}{1.0}

\subsection{Experimental Settings}
\subsubsection{Datasets}

We conduct our experiments on three public datasets. In all datasets, we randomly split them into 8:1:1 as the training, validation and test set respectively. 
We specify the dataset preprocess procedure in Appendix \ref{sec:dataset-preprocess}. Table \ref{tab:dataset-stat} provides the core statistics of pre-processed datasets.

\subsubsection{Metrics}
We evaluate all models using two commonly used evaluation metric in CTR prediction community \cite{barsctr2021}: LogLoss and AUC (Area under the ROC curve). 
In the CTR problem, a difference of \textbf{0.001} in AUC is generally considered significant \cite{tran2024thorough,optembed2022}.
The lower LogLoss suggests better performance, while the higher AUC implies more accurate recommendations.
The sparsity rate $t$ indicates the portion of embedding parameters that have been removed compared to the original embedding table, defined as $t = \len{\mathcal{Q}}/{nd}$.

\subsubsection{Implementation Details}
All methods are tested with two well-known CTR backbones: DeepFM \cite{deepfm2017} and DCN--Mix \cite{dcnv2}. We implement two variants of Shaver, namely Shaver-Zero that uses the traditional zero-padding on pruned parameters, and Shaver-Codebook described in Section \ref{sec:codebook}.
Descriptions of all baselines are provided in Appendix \ref{sec:baselines}.
We adhere strictly to the experimental settings described in \cite{tran2024thorough} for reproducibility, adopting the same data splits and hyperparameter search ranges.
To ensure a fair comparison, we applied the same hyperparameters used in the training of the original model during our fine-tuning step. We used both the training and validation datasets to calculate the attribution scores. 
We select final checkpoints according to validation AUC.

\renewcommand{\arraystretch}{0.8}
\setlength\tabcolsep{3.5pt}
\begin{table*}[pth]
    \centering
        \caption{Comparative results with non-single-shot baselines under three sparsity rates. All methods adopt DCN-Mix as the CRS backbone, where we use $\uparrow$ and $\downarrow$ to respectively mark higher-is-better and lower-is-better metrics. For OptEmb, we did not report results on some settings where it fails to converge. ``\#Params'' indicates the parameter size of the embedding table (in millions). In each dataset, the best result is marked in bold and the second best one is underlined.}
    \vspace{-.3cm}
\begin{tabular}{llcccccccccc}
    \toprule
    \multirow[c]{2}{*}{Dataset} & \multirow[c]{2}{*}{Method} &Single- & \multicolumn{3}{c}{$t$ = 50\%} & \multicolumn{3}{c}{$t$ = 80\%} & \multicolumn{3}{c}{$t$ = 95\%} \\
    \cmidrule(lr){4-6}\cmidrule(lr){7-9}\cmidrule(lr){10-12}
     & & shot? & AUC$\uparrow$ & LogLoss$\downarrow$ & \#Params & AUC$\uparrow$ & LogLoss$\downarrow$ & \#Params & AUC$\uparrow$ & LogLoss$\downarrow$ & \#Params \\
    \midrule

\multirow[c]{8}{*}{Criteo} 
     & QR \cite{qr2019} & \xmark & 0.8095 & 0.4425 & 8.69M & 0.8091 & 0.4427 & 3.48M & 0.8064 & 0.4452 & 870K \\
     & TTRec \cite{ttrec2021} & \xmark & 0.8112 & \underline{0.4408} & 8.26M & 0.8103 & 0.4416 & 3.48M & 0.8106 & 0.4414 & 870K \\
     & PEP \cite{pep2021} & \xmark & \underline{0.8113} & \textbf{0.4407} & 8.69M & 0.8108 & 0.4413 & 3.46M & \textbf{0.8114} & \textbf{0.4406} & 831K \\
     & OptEmb \cite{optembed2022} & \xmark & \multicolumn{1}{c}{-} & \multicolumn{1}{c}{-} & \multicolumn{1}{c}{-} & \underline{0.8112} & \underline{0.4408} & 1.45M & 0.8109 & \underline{0.4408} & 856K \\
     & CERP \cite{cerp2023} & \xmark & 0.8110 & 0.4410 & 8.66M & 0.8106 & 0.4416 & 3.44M & 0.8107 & 0.4414 & 850K \\
     \cmidrule(lr){2-12}
     & Shaver-Zero & \cmark & \textbf{0.8114} & \textbf{0.4407} & 8.69M & \textbf{0.8114} & \textbf{0.4407} & 3.48M & 0.8110 & 0.4412 & 869K \\
      & Shaver-Codebook & \cmark & \textbf{0.8114} & \textbf{0.4407} & 8.69M & \textbf{0.8114} & \textbf{0.4407} & 3.48M & \underline{0.8112} & 0.4410 & 870K \\

\midrule

\multirow[c]{8}{*}{Avazu} 
   & QR \cite{qr2019} & \xmark & 0.7743 & 0.3857 & 35.43M & 0.7698 & 0.3922 & 14.17M & 0.7652 & 0.3895 & 3.54M \\
 & TTRec \cite{ttrec2021} & \xmark & 0.7711 & 0.3873 & 35.43M & 0.7710 & 0.3900 & 14.17M & 0.7678 & 0.3903 & 3.54M \\
 & PEP \cite{pep2021} & \xmark & 0.7754 & 0.3852 & 33.55M & 0.7669 & 0.3906 & 12.57M & 0.7564 & 0.3938 & 3.52M \\
 & OptEmb \cite{optembed2022} & \xmark & \multicolumn{1}{c}{-} & \multicolumn{1}{c}{-} & \multicolumn{1}{c}{-} & 0.7663 & 0.3877 & 13.24M & \underline{0.7683} & 0.3895 & 3.43M \\
 & CERP \cite{cerp2023} & \xmark & 0.7723 & 0.3852 & 35.34M & 0.7673 & 0.3873 & 13.05M & 0.7631 & 0.3933 & 3.54M \\
 
\cmidrule(lr){2-12}

 & Shaver-Zero & \cmark & \underline{0.7759} & \textbf{0.3837} & 35.43M & \underline{0.7745} & \textbf{0.3840} & 14.17M & \textbf{0.7690} & \textbf{0.3865} & 3.54M \\
  & Shaver-Codebook & \cmark & \textbf{0.7762} & \underline{0.3838} & 35.43M & \textbf{0.7749} & \underline{0.3842} & 14.17M & \underline{0.7683} & \underline{0.3873} & 3.54M \\

\midrule

\multirow[c]{8}{*}{KDD} 
& QR \cite{qr2019} & \xmark & 0.7724 & 0.1581 & 48.16M & 0.7709 & 0.1584 & 19.26M & 0.7702 & 0.1588 & 4.82M \\
 & TTRec \cite{ttrec2021} & \xmark & 0.7862 & 0.1576 & 48.16M & \underline{0.7843} & \underline{0.1564} & 19.26M & 0.7815 & \textbf{0.1559} & 4.82M \\
 & PEP \cite{pep2021} & \xmark & 0.7754 & 0.1576 & 45.93M & 0.7731 & 0.1581 & 18.88M & 0.7737 & 0.1579 & 4.50M \\
 & OptEmb \cite{optembed2022} & \xmark & \textbf{0.7877} & \textbf{0.1558} & 48.99M  
 & 0.7818 & 0.1593 & 15.10M & \underline{0.7826} & 0.1591 & 4.80M\\
 & CERP \cite{cerp2023} & \xmark & 0.7841 & 0.1566 & 47.50M & 0.7781 & 0.1573 & 19.18M & 0.7721 & 0.1589 & 4.68M \\
 
\cmidrule(lr){2-12}

 & Shaver-Zero & \cmark & 0.7858 & 0.1563 & 48.16M & 0.7826 & 0.1569 & 19.26M & 0.7699 & 0.1601 & 4.82M \\
  & Shaver-Codebook & \cmark & \underline{0.7870} & \underline{0.1559} & 48.16M & \textbf{0.7870} & \textbf{0.1557} & 19.26M & \textbf{0.7836} & \underline{0.1563} & 4.82M \\
\bottomrule
\end{tabular}
    \vspace{-.3cm}

\label{tab:RQ1}
\end{table*}
\renewcommand{\arraystretch}{1.0}

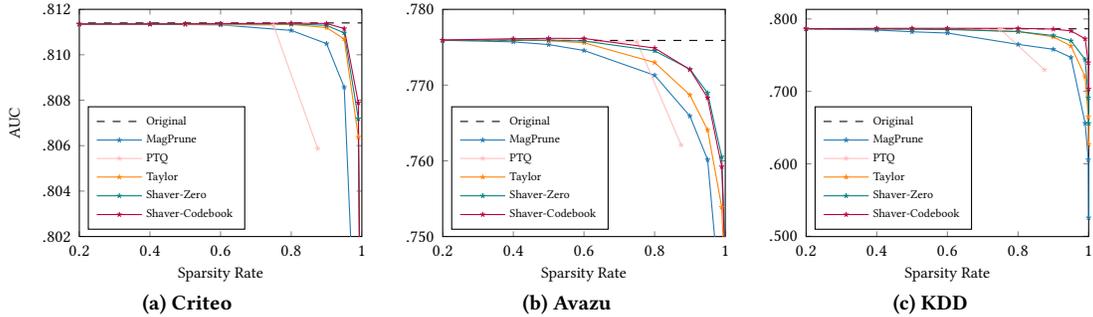
\begin{figure*}[t]
    \centering
    \subcaptionbox{Criteo}
    {
        \vspace{-.2cm}
        \pgfplotstableread[col sep=comma]{figures/rq2/data/criteo_dcn.csv}\criteodcn
\definecolor{blue1}{rgb}{0.121569,0.466667,0.705882}%
\begin{tikzpicture}
\tikzstyle{every node}=[font=\scriptsize]

\begin{axis}[
    xlabel={Sparsity Rate},
    ylabel={AUC},
    width=0.3\textwidth,
    xmin=0.2, xmax=1.0,
    ymin=0.802, ymax=0.812,
    xtick={0.2, 0.4, 0.6, 0.8, 1.0},
    y label style={at={(axis description cs:-0.18,.5)}},
    x label style={at={(axis description cs:0.5,-0.1)}},
    yticklabel style={
        /pgf/number format/precision=3,
        /pgf/number format/fixed zerofill,
        /pgf/number format/skip 0.=true
    },
    every x tick label/.append style={font=\scriptsize},
    legend pos=south west,
    legend style={nodes={scale=0.75}},
    ymajorgrids=false,
    grid style=dashed,
    tickwidth=0.07cm,
    every axis plot/.append style={
        mark=star,
        mark options={scale=0.5},
    },
    legend cell align=left,
]
\legend{Original, MagPrune, PTQ, Taylor, Shaver-Zero, Shaver-Codebook}

\addplot[
    color=black,
    dashed,
    mark=none,
] coordinates {
   (0.0,0.8114) (1.0,  0.8114)
};

\addplot[
    color=blue1,
] table[
    x=sparse-rate, 
    y=mag
]
{\criteodcn};

\addplot[
    color=pink,
] table[
    x=sparse-rate, 
    y=ptq
]
{\criteodcn};

\addplot[
    color=orange,
] table[
    x=sparse-rate, 
    y=taylor.pth
]
{\criteodcn};

\addplot[
    color=teal,
] table[
    x=sparse-rate, 
    y=zero.pth
]
{\criteodcn};

\addplot[
    color=purple,
] table[
    x=sparse-rate, 
    y=codebook.pth
]
{\criteodcn};

\end{axis}
\end{tikzpicture}
    }
    \subcaptionbox{Avazu}
    {
        \vspace{-.2cm}
        \pgfplotstableread[col sep=comma]{figures/rq2/data/avazu_dcn.csv}\avazudcn
\definecolor{blue1}{rgb}{0.121569,0.466667,0.705882}%
\begin{tikzpicture}
\tikzstyle{every node}=[font=\scriptsize]

\begin{axis}[
    xlabel={Sparsity Rate},
    width=0.3\textwidth,
    xmin=0.2, xmax=1.0,
    ymin=0.75, ymax=0.78,
    xtick={0.2, 0.4, 0.6, 0.8, 1.0},
    y label style={at={(axis description cs:-0.18,.5)}},
    x label style={at={(axis description cs:0.5,-0.1)}},
    yticklabel style={
        /pgf/number format/precision=3,
        /pgf/number format/fixed zerofill,
        /pgf/number format/skip 0.=true
    },
    every x tick label/.append style={font=\scriptsize},
    legend pos=south west,
    legend style={nodes={scale=0.75}},
    ymajorgrids=false,
    grid style=dashed,
    tickwidth=0.07cm,
    every axis plot/.append style={
        mark=star,
        mark options={scale=0.5},
    },
    legend cell align=left,
]
\legend{Original, MagPrune, PTQ, Taylor, Shaver-Zero, Shaver-Codebook}

\addplot[
    color=black,
    dashed,
    mark=none,
] coordinates {
   (0.0, 0.7759) (1.0,  0.7759)
};

\addplot[
    color=blue1,
] table[
    x=sparse-rate, 
    y=mag
]
{\avazudcn};

\addplot[
    color=pink,
] table[
    x=sparse-rate, 
    y=ptq
]
{\avazudcn};

\addplot[
    color=orange,
] table[
    x=sparse-rate, 
    y=taylor.pth
]
{\avazudcn};

\addplot[
    color=teal,
] table[
    x=sparse-rate, 
    y=zero.pth
]
{\avazudcn};

\addplot[
    color=purple,
] table[
    x=sparse-rate, 
    y=codebook.pth
]
{\avazudcn};

\end{axis}
\end{tikzpicture}
    }
    \subcaptionbox{KDD}
    {
        \vspace{-.2cm}
        \pgfplotstableread[col sep=comma]{figures/rq2/data/kdd_dcn.csv}\kdddcn
\definecolor{blue1}{rgb}{0.121569,0.466667,0.705882}%
\begin{tikzpicture}
\tikzstyle{every node}=[font=\scriptsize]

\begin{axis}[
    xlabel={Sparsity Rate},
    width=0.3\textwidth,
    xmin=0.2, xmax=1.0,
    xtick={0.2, 0.4, 0.6, 0.8, 1.0},
    y label style={at={(axis description cs:-0.18,.5)}},
    x label style={at={(axis description cs:0.5,-0.1)}},
    yticklabel style={
        /pgf/number format/precision=3,
        /pgf/number format/fixed zerofill,
        /pgf/number format/skip 0.=true
    },
    every x tick label/.append style={font=\scriptsize},
    legend pos=south west,
    legend style={nodes={scale=0.75}},
    ymajorgrids=false,
    grid style=dashed,
    tickwidth=0.07cm,
    every axis plot/.append style={
        mark=star,
        mark options={scale=0.5},
    },
    legend cell align=left,
]
\legend{Original, MagPrune, PTQ, Taylor, Shaver-Zero, Shaver-Codebook}

\addplot[
    color=black,
    dashed,
    mark=none,
] coordinates {
   (0.0, 0.7864) (1.0,  0.7864)
};

\addplot[
    color=blue1,
] table[
    x=sparse-rate, 
    y=mag
]
{\kdddcn};

\addplot[
    color=pink,
] table[
    x=sparse-rate, 
    y=ptq
]
{\kdddcn};

\addplot[
    color=orange,
] table[
    x=sparse-rate, 
    y=taylor.pth
]
{\kdddcn};

\addplot[
    color=teal,
] table[
    x=sparse-rate, 
    y=zero.pth
]
{\kdddcn};

\addplot[
    color=purple,
] table[
    x=sparse-rate, 
    y=codebook.pth
]
{\kdddcn};

\end{axis}
\end{tikzpicture}
    }
    \vspace{-.2cm}
    \caption{Comparative results with single-shot baselines, where DCN-Mix is used as the backbone.}
    \label{fig:rq2}
\end{figure*}

\subsection{Comparison with Non-single-shot Pruning Methods (RQ1)}

Table \ref{tab:RQ1} shows the performance of Shaver and other non-single-shot pruning methods with the DCN-Mix backbone. The baselines compared are QR \cite{qr2019}, TTRec \cite{ttrec2021}, PEP \cite{pep2021}, OptEmb \cite{optembed2022}, and CERP \cite{cerp2023}, whose detailed descriptions are provided in Appendix \ref{sec:baselines}. As DCN-Mix is generally a more performant backbone compared with DeepFM, results with the DeepFM backbone and the inference efficiency are deferred to Appendix \ref{sec:more-exp} due to page limits. For this comparison, we test three sparsity rate settings with $t=\{50\%, 80\%, 95\%\}$ as all baselines under this category needs to be retrained for every setting.
%
Across three datasets, Shaver achieves competitive results with other non-single-shot pruning methods that require costly training steps for each parameter budget.
Shaver demonstrates a strong performance compared to other methods, especially with an 80\% sparsity rate. 
While in the low sparsity rate, the performance gap between methods is small as it is generally easier to reduce parameters. 
At 95\% sparsity, Shaver delivers impressive results exceeding other methods, despite a lesser extent. 
Higher sparsity rates pose challenges for single-shot pruning due to removing a larger portion of original model parameters, thus, leading to a more significant model degradation.

\vspace{-0.3cm}
\subsection{Comparison with Single-shot Pruning Methods (RQ2)}

\begin{figure*}[!th]
\begin{minipage}[t]{1.5\columnwidth}
    \centering
    \subcaptionbox{Criteo}
    {
        \vspace{-.2cm}
        \pgfplotstableread[col sep=comma]{figures/rq3/data/criteo_deepfm.csv}\criteotable
\definecolor{blue1}{rgb}{0.121569,0.466667,0.705882}%
\begin{tikzpicture}
\tikzstyle{every node}=[font=\scriptsize]

\begin{axis}[
    xlabel={Sparsity Rate},
    ylabel={AUC},
    width=0.365\linewidth,
    xmin=0.2, xmax=1.0,
    ymin=0.8, ymax=0.812,
    xtick={0.2, 0.4, 0.6, 0.8, 1.0},
    ytick={0.798,0.802,...,0.812},
    y label style={at={(axis description cs:-0.18,.5)}},
    x label style={at={(axis description cs:0.5,-0.1)}},
    yticklabel style={
        /pgf/number format/precision=3,
        /pgf/number format/fixed zerofill,
        /pgf/number format/skip 0.=true
    },
    every x tick label/.append style={font=\scriptsize},
    legend pos=south west,
    legend style={nodes={scale=0.75}},
    ymajorgrids=false,
    grid style=dashed,
    tickwidth=0.07cm,
    every axis plot/.append style={
        mark=star,
        mark options={scale=0.5},
    },
]
\legend{$p=10\%$, $p=30\%$, $p=64\%$, $p=100\%$}

\addplot[
    color=blue1,
] table[
    x=sparse-rate, 
    y=codebook-0.0.pth
]
{\criteotable};

\addplot[
    color=orange,
] table[
    x=sparse-rate, 
    y=codebook-0.2.pth
]
{\criteotable};

\addplot[
    color=teal,
] table[
    x=sparse-rate, 
    y=codebook-0.6.pth
]
{\criteotable};

\addplot[
    color=red,
] table[
    x=sparse-rate, 
    y=codebook-1.0.pth
]
{\criteotable};

\end{axis}
\end{tikzpicture}
    }\hfill
    \subcaptionbox{Avazu}
    {
        \vspace{-.2cm}

\pgfplotstableread[col sep=comma]{figures/rq3/data/avazu_deepfm.csv}\avazudeepfm
\definecolor{blue1}{rgb}{0.121569,0.466667,0.705882}%
\begin{tikzpicture}
\tikzstyle{every node}=[font=\scriptsize]

\begin{axis}[
    xlabel={Sparsity Rate},
    width=0.365\linewidth,
    xmin=0.2, xmax=1.0,
    ymin=0.7, ymax=0.77,
    xtick={0.2, 0.4, 0.6, 0.8, 1.0},
    y label style={at={(axis description cs:-0.18,.5)}},
    x label style={at={(axis description cs:0.5,-0.1)}},
    yticklabel style={
        /pgf/number format/precision=3,
        /pgf/number format/fixed zerofill,
        /pgf/number format/skip 0.=true
    },
    every x tick label/.append style={font=\scriptsize},
    legend pos=south west,
    legend style={nodes={scale=0.75}},
    ymajorgrids=false,
    grid style=dashed,
    tickwidth=0.07cm,
    every axis plot/.append style={
        mark=star,
        mark options={scale=0.5},
    },
]
\legend{$p=10\%$, $p=30\%$, $p=64\%$, $p=100\%$}

\addplot[
    color=blue1,
] table[
    x=sparse-rate, 
    y=codebook-0.0.pth
]
{\avazudeepfm};

\addplot[
    color=orange,
] table[
    x=sparse-rate, 
    y=codebook-0.2.pth
]
{\avazudeepfm};

\addplot[
    color=teal,
] table[
    x=sparse-rate, 
    y=codebook-0.6.pth
]
{\avazudeepfm};

\addplot[
    color=red,
] table[
    x=sparse-rate, 
    y=codebook-1.0.pth
]
{\avazudeepfm};

\end{axis}
\end{tikzpicture}
    }\hfill
    \vspace{-.2cm}
    \subcaptionbox{KDD}
    {
        \vspace{-.2cm}

\pgfplotstableread[col sep=comma]{figures/rq3/data/kdd_deepfm.csv}\criteotable
\definecolor{blue1}{rgb}{0.121569,0.466667,0.705882}%
\begin{tikzpicture}
\tikzstyle{every node}=[font=\scriptsize]

\begin{axis}[
    xlabel={Sparsity Rate},
    width=0.365\linewidth,
    xmin=0.2, xmax=1.0,
    ymin=0.7, ymax=0.79,
    xtick={0.2, 0.4, 0.6, 0.8, 1.0},
    y label style={at={(axis description cs:-0.18,.5)}},
    x label style={at={(axis description cs:0.5,-0.1)}},
    yticklabel style={
        /pgf/number format/precision=3,
        /pgf/number format/fixed zerofill,
        /pgf/number format/skip 0.=true
    },
    every x tick label/.append style={font=\scriptsize},
    legend pos=south west,
    legend style={nodes={scale=0.75}},
    tickwidth=0.07cm,
    ymajorgrids=false,
    grid style=dashed,
    every axis plot/.append style={
        mark=star,
        mark options={scale=0.5},
    },
]
\legend{$p=10\%$, $p=30\%$, $p=64\%$, $p=100\%$}

\addplot[
    color=blue1,
] table[
    x=sparse-rate, 
    y=codebook-0.0.pth
]
{\criteotable};

\addplot[
    color=orange,
] table[
    x=sparse-rate, 
    y=codebook-0.2.pth
]
{\criteotable};

\addplot[
    color=teal,
] table[
    x=sparse-rate, 
    y=codebook-0.6.pth
]
{\criteotable};

\addplot[
    color=red,
] table[
    x=sparse-rate, 
    y=codebook-1.0.pth
]
{\criteotable};

\end{axis}
\end{tikzpicture}
    }\hfill
    \vspace{-.1cm}
    \caption{Trade-off between performance and data size used to estimate Shapley value. $p$ is the portion of the full dataset used to calculate the Shapley value.}
    \label{fig:tradeoff-dataset-size}
\end{minipage}\hfill 
\begin{minipage}[t]{0.5\columnwidth}
\includegraphics[width=0.9\textwidth]{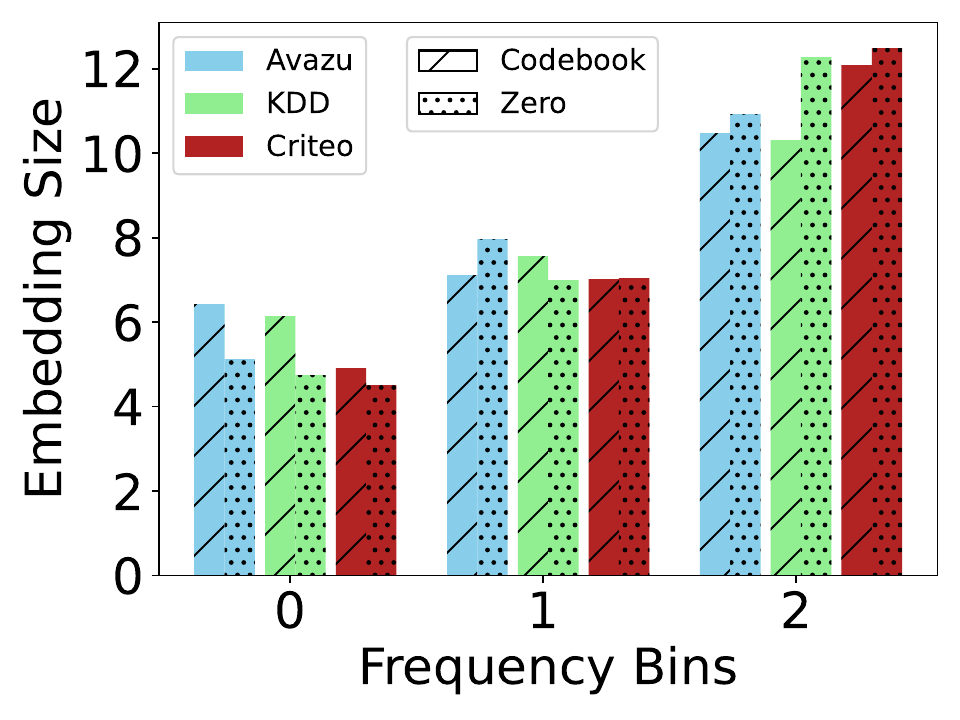} 
\vspace{-.4cm}
\caption{Comparison of how different Shaver variants prune embeddings for features in different frequency buckets.}
 \label{fig:rq4}
\end{minipage}\hfill
\end{figure*}

To make further comparisons, we evaluate our approach against other single-shot compression methods. The baselines compared are MagPrune \cite{tran2024thorough}, PTQ \cite{ptq2021}, and Taylor \cite{taylor2017} (see Appendix \ref{sec:baselines} for descriptions). Same with RQ1, DCN-Mix backbone is adopted by default, where results with DeepFM backbone are provided in Appendix \ref{sec:more-exp}. Notably, different from RQ1, as all baselines tested here are all single-shot methods, we have been able to test with a wider variety of sparsity rates 
$t=\{0.2, 0.4, 0.5, 0.6, 0.8, 0.9, 0.95, 0.99, 0.999\}$.

Figure \ref{fig:rq2} shows the results.
Shaver-Codebook is the most competitive method for KDD and Criteo, retaining the performance even with only 10\% and 5\% parameters of the original models, respectively.
The second best is Shaver-Zero, whose performance is similar to Shaver-Codebook in the Avazu dataset. 
These results indicate promising potential in applying the Shapley value in compressing models for CTR problems.
Taylor is the third best as they also utilize information from the training and validation datasets, especially in the Criteo dataset, where the difference is only noticeable under high-sparsity settings.
PTQ provides strong performance initially but could not hold up at a higher sparsity rate, as it allocates the same memory budget for all features, limiting their performance.

\subsection{Effect of Post-Pruning Finetuning (RQ3)}

\begin{table}[thb]
    \centering
    \renewcommand{\arraystretch}{0.8}
\setlength\tabcolsep{1.5pt}
    \caption{Effect of post-pruning finetuning, where $\dagger$ indicates variants that are finetuned with the pruned embeddings. The results are reported in AUC.}
    \vspace{-.3cm}
    \begin{tabular}{llccc}
    \toprule
    {Dataset} & {Method} & {$t$ = 50\%} & {$t$ = 80\%} & {$t$ = 95\%} \\
    \midrule
    \multirow[c]{4}{*}{Criteo} 
     & Shaver-Zero & 0.8114 & 0.8114 & 0.8110 \\
     & Shaver-Zero$\dagger$ & 0.8116 & 0.8114 &  0.8112 \\
      & Shaver-Codebook & 0.8114 & 0.8114 & 0.8112 \\
     & Shaver-Codebook$\dagger$ & 0.8115 & 0.8114 & 0.8112 \\
     \midrule
     
     \multirow[c]{4}{*}{Avazu} 
     & Shaver-Zero & 0.7759 & 0.7745 & 0.7690 \\
 & Shaver-Zero$\dagger$ & 0.7777 & 0.7735 & 0.7703 \\
  & Shaver-Codebook & 0.7762 & 0.7749 & 0.7683 \\
 & Shaver-Codebook$\dagger$ & 0.7781 & 0.7717 & 0.7670 \\
     \midrule
     \multirow[c]{4}{*}{KDD} 
     & Shaver-Zero & 0.7858 & 0.7826 & 0.7699 \\
 & Shaver-Zero$\dagger$ & 0.7858 & 0.7826 & 0.7775 \\
  & Shaver-Codebook & 0.7870 & 0.7870 & 0.7836 \\
 & Shaver-Codebook$\dagger$ & 0.7754 &  0.7806 & 0.7836 \\
    \bottomrule
    \end{tabular}
    \label{tab:finetune-exp}
\end{table}

To answer RQ3, Table \ref{tab:finetune-exp} shows the effect of fine-tuning pruned models. This experiment also mimics scenarios where a deployed recommendation model has sufficient computing resources to perform further on-device updates. While there is some variation in results across different datasets, in most cases, fine-tuning models further enhance their performance or maintain it at its current level, or with minimal loss. However, occasionally, fine-tuning can lead to decreased model performance, particularly for some settings in KDD.
The reason for the performance drop is that Shaver uses the validation set to estimate Shapley values, which can cause data leakage and potentially degrade the final checkpoint performance post-fine-tuning. 
Moreover, fine-tuning seems to have a larger effect on high sparsity rates with Shaver-Zero, where performance degradation is more pronounced.

\subsection{Performance and Data Size Trade-off (RQ4)}

Figure \ref{fig:tradeoff-dataset-size} depicts the trade-off between the performance and the proportion $p$ of the full dataset $\mathcal{D}$ employed. Note that $\mathcal{D}$ combines both training and validation samples in our setting. Under all configurations, we guarantee the algorithms have converged and seen all data. Empirically, the required runtime is linearly scaled with the dataset size employed, as Shaver generally converges after iterating through the whole dataset once.

In the Criteo and KDD datasets, the amount of data employed has a negligible effect on model performance. Conversely, in the Avazu dataset, the volume of data has a markedly greater impact. One possible hypothesis for this discrepancy is that despite containing more features than the Criteo dataset, the Avazu dataset has fewer instances. Consequently, the validation dataset is insufficient to represent the actual data distribution. 
Nonetheless, as our methods only require forward passes, we can effortlessly parallelize the computation compared to other training-required methods. Moreover, the overhead remains minimal, even with the entire Criteo dataset -- our worst case as our computation cost increases with the number of fields -- Shaver takes approximately 2.5 hours to compute the Shapley value. 
In comparison, for any given parameter budget, training each method generally takes over 1 hour, without including the time needed for hyperparameter tuning.

\subsection{Qualitative Analysis (RQ5)}
We aim to unveil the key difference between the traditional zero-out approach and our codebook-based approach in compressed models.
Therefore, employing DCN-Mix with a 50\% sparsity rate, we split each dataset's feature into three different bins based on frequency so that each bin had the same number of features. Figure \ref{fig:rq4} shows the average embedding size assigned per feature in each group.

The codebook approach assigns less parameter budget to high frequency. 
One plausible explanation is that the codebook already contained parameters representing the high-frequency features.
Thus, the model could have a larger budget to allocate to the lower frequency features, improving the model's performance.

\section{Conclusions}

In this paper, we propose Shaver, which employs Shapley value to fairly distribute each embedding parameter contribution and effectively prune the embedding table in CTR models. Shaver first provide an efficient method to calculate the Shapley value, and then implement a novel field-aware codebook quantization to represent the removed parameters. Extensive experiments show that Shaver achieves competitive performance without expensive retraining and fine-tuning steps for each parameter budget.

\begin{acks}
This work is supported by Australian Research Council under the streams of Future Fellowship (No. FT210100624), Linkage Project (No. LP230200892), Discovery Early Career Researcher Award (No. DE230101033), and Discovery Project (No. DP240101108 and No. DP240101814).
\end{acks}

\bibliographystyle{ACM-Reference-Format}
\balance
\bibliography{00_main}

\appendix

\section{Proofs}
\label{sec:proofs}

\begin{figure}[htbp]
    \centering
    \caption{A graphic explanation on the notion of corresponding player.}
    \includegraphics[width=0.45\textwidth]{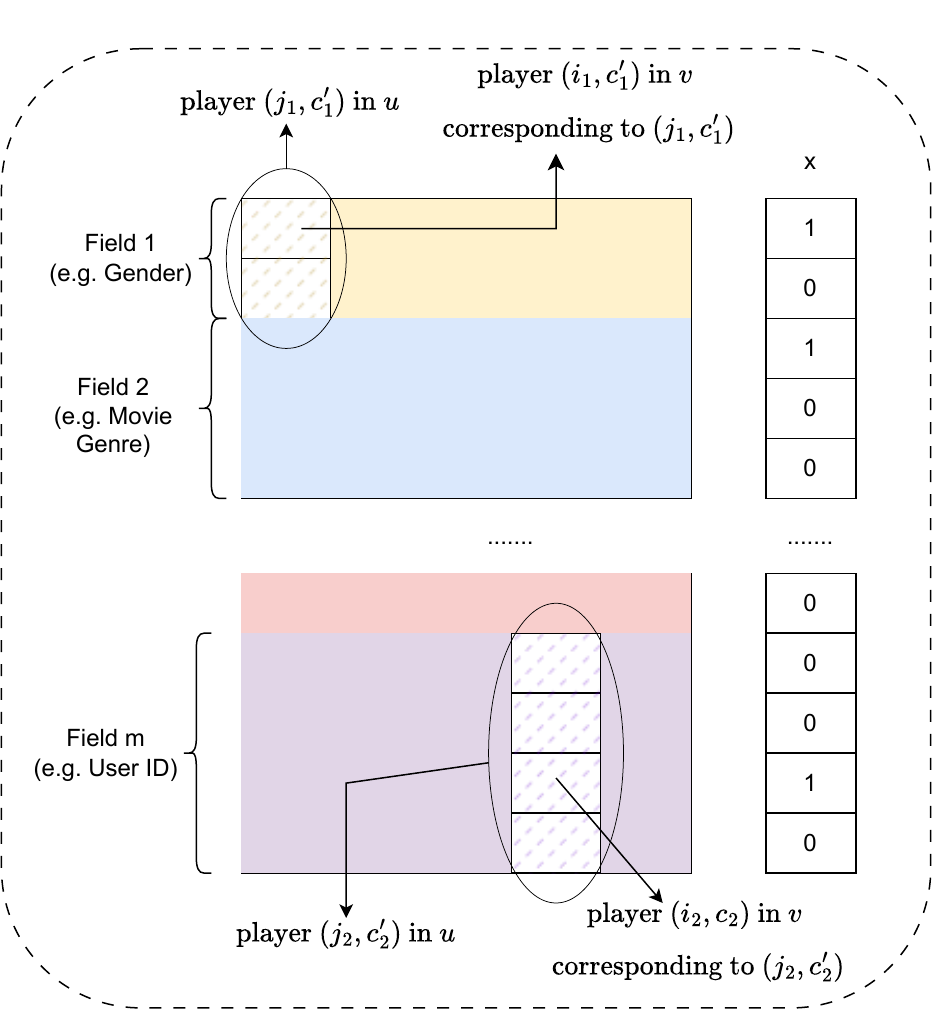}
    \label{fig:notation}
\end{figure}

\localeq*
\begin{proof}

Figure \ref{fig:notation} shows two examples of corresponding players between game $v$ and $u$.
In the proof below, we first simplify the notations by using $i$ and $j$ to represent the player $(i, c)$ and $(j, c)$.
Let's define function $g: \{0,1\}^{\len{\mathcal{N}_v}} \rightarrow \{ 0, 1 \}^{\len{\mathcal{N}_u}}$ as follow:
\begin{equation}
    g(\mathcal{S}_v) [j] =\begin{cases}
        1 \text{ if corresponding } i \in \mathcal{S}_v \\
        0 \text{ else}
    \end{cases},
\end{equation}
with $\mathcal{S}_v$  is a subset of players in game $v$.  
In other words, 
$g(\mathcal{S}_v)=\mathcal{S}_u$ if and only if $\forall j \in \mathcal{S}_u, \text{corresponding } i \in \mathcal{S}_v$ and $\forall j \in \mathcal{N}_u \setminus \mathcal{S}_u, \text{corresponding } i \notin \mathcal{S}_v$.
Intuitively, this $g$ function provide a method to map from a set $S_v$ in $v$ to the corresponding set $S_u$ in $u$. 
Let's define marginal contribution $\Delta$ as:
\begin{equation}
\begin{split}
    \Delta(i | \mathcal{S}_v) &= v(\mathcal{S} \cup i) - v(\mathcal{S}). \\    
\end{split}
\end{equation}
It is easy to show that if $g(\{ i \}) = \{ j \}$ and $g(\mathcal{S}_v) = \mathcal{S}_u$ then:
$$g(\mathcal{S}_v \cup \{ i \}) = \mathcal{S}_u \cup \{ j \}.$$
Consequently, if $g(\mathcal{S}_v)=\mathcal{S}_u \text{ and } g(\{ i \}) = \{ j \}
$, then:
\begin{equation*}
    \Delta(i|{\mathcal{S}_v})=\Delta(j|{\mathcal{S}_u}).
\end{equation*}

It is trivial to show that $\phi_v^{x,y}(i)=0$ if $g( \{ i \} ) = \emptyset$ (feature $i$ is not activated). 
We will only consider case that $\exists j \in \mathcal{N}_u, j = g( \{ i \} )$.
By definition, we have Shapley value of a player $i$ in game $v$ as:
\begin{equation}
    \phi_v^{x,y} (i) = \frac{1}{\len{\mathcal{N}_v}} \sum_{\mathcal{S}_v \subseteq \mathcal{N}_v \setminus i } \binom{\len{\mathcal{N}_v} - 1}{\len{\mathcal{S}_v}}^{-1} \Delta(i|{\mathcal{S}_v}).
    \label{eq:phi-v-ori}
\end{equation}

It is easy to show that:
\begin{equation}
    \len{\mathcal{N}_v} = nd; \len{\mathcal{N}_u} = md.
\end{equation}

With known $\mathcal{S}_u$ and $\left| \mathcal{S}_v \right|$, the number of $\mathcal{S}_v$ s.t. $g(\mathcal{S}_v) = \mathcal{S}_u$ is 

$$\binom{\len{\mathcal{N}_v} - \len{\mathcal{N}_u}}{\len{\mathcal{S}_v} - \len{\mathcal{S}_u}} = \binom{nd -md}{\len{\mathcal{S}_v} - \len{\mathcal{S}_u}},$$

as we only need to find $\len{\mathcal{S}_v} - \len{\mathcal{S}_u}$ from a pool of $\len{\mathcal{N}_v} - \len{\mathcal{N}_u}$ elements.
Thus, we can write Eq. \ref{eq:phi-v-ori} as:

\begin{equation}
\begin{split}
    \phi_v^{x,y} (i) \\
    = \frac{1}{nd} 
        \sum_{\mathcal{S}_u \subseteq \mathcal{N}_u \setminus j} \Delta(j|{\mathcal{S}_u})
        \sum_{k=\len{\mathcal{S}_u}}^{nd - md + \len{\mathcal{S}_u}} \binom{nd - md}{k - \len{\mathcal{S}_u}} {\binom{nd - 1}{k}}^{-1} \\
    = \frac{1}{nd} \sum_{\mathcal{S}_u \subseteq \mathcal{N}_u \setminus j} \Delta(j|{\mathcal{S}_u}) \sum_{k=0}^{nd - md} \binom{nd - md}{k} {\binom{nd - 1}{k + \len{\mathcal{S}_u}}}^{-1}. \\
\end{split}
\label{eq:phi-i-2}
\end{equation}

With some rearrangements, we have the following equation:

\begin{equation}
   \begin{split}
          \binom{nd - md}{k} \binom{md - 1}{ \len{\mathcal{S}_u}} {\binom{nd - 1}{k + \len{\mathcal{S}_u}}}^{-1} \\
          = \frac{nd}{md} \binom{k + \len{\mathcal{S}_u}}{\len{\mathcal{S}_u}} \binom{nd - 1 - k - \len{\mathcal{S}_u}}{md - \len{\mathcal{S}_u} - 1} \binom{nd}{md}^{-1}.
   \end{split}
   \label{eq:arrangement}
\end{equation}

And by Chu–Vandermonde Identity:
\begin{equation}
    \sum_{k=0}^{nd - md} \binom{k + \len{\mathcal{S}_u}}{\len{\mathcal{S}_u}} \binom{nd - 1 - k - \len{\mathcal{S}_u}}{md - \len{\mathcal{S}_u} - 1} =  \binom{nd}{md}.
    \label{eq:Chu–Vandermonde}
\end{equation}

With Eq. (\ref{eq:phi-i-2}), (\ref{eq:arrangement}), and (\ref{eq:Chu–Vandermonde}), we have:

\begin{equation}
\begin{split}
    \phi_v^{x,y}  (i) &= \frac{1}{nd} \sum_{\mathcal{S}_u \subseteq \mathcal{N}_u \setminus j} \Delta(j|{\mathcal{S}_u}) \frac{nd}{md} \binom{md - 1}{\len{\mathcal{S}_u}}^{-1} \\
    &= \phi_u^{x,y}  (j).
\end{split}
\end{equation}
\end{proof}

In the part below, we provide proof for the codebook optimization result:
\begin{equation}
\label{eq:codebook-target}
\begin{split}
    \mathbf{C}^{*} &= \min_{\mathbf{C}} \meanE_{\mathbf{x} \sim \mathcal{D}, \len{\mathcal{{Q}}}=B} \left[ \lVert \mathbf{E}^{\top}\mathbf{x} - \mathbf{E}_{\mathcal{Q},\mathbf{C}}^{\top}\mathbf{x} \rVert_2^2 \right],\\
    \mathbf{C}^{*}[j,:] &= \frac{\sum_{i \in \mathcal{F}_j} p_i \mathbf{V}_j[i,:]}{\sum_{i \in \mathcal{F}_j} p_i}, \\
\end{split}
\end{equation}

with $B$ is any given memory budget, $\mathcal{F}_j$ is the set of features that belong to field $j$, $p_i$ is the appear frequency of feature $i$. To begin with, we have the following equation:

\begin{equation}
\begin{split}
\meanE_{\mathbf{x} \sim \mathcal{D}, \len{\mathcal{{Q}}}=B} \left[ \lVert \mathbf{E}^{\top}\mathbf{x} - \mathbf{E}_{\mathcal{Q},\mathbf{C}}^{\top}\mathbf{x} \rVert_2^2 \right] \\
= \sum_{j=1}^{m} \sum_{i \in \mathcal{F}_j} \sum_{c=1}^{d}
P\left( (i, c) \in \mathcal{Q} \big{\vert} \len{\mathcal{Q}} = B \right) 
p_i ( \mathbf{E}[i,c] - \mathbf{C}[j,c] )^2,
\end{split}
\end{equation}
where $P\left( (i, c) \in \mathcal{Q} \big{\vert} \len{\mathcal{Q}} = B \right)$ is the probability that parameter $(i, c)$ is in $\mathcal{Q}$. Since $\mathcal{Q}$ is sampled from a uniform distribution with a budget $B$, the probability of any embedding parameter appearing in $\mathcal{Q}$ is uniform across all parameters and can be treated as a constant.
By removing the above constant, we have the following equation:
\begin{equation}
    \mathbf{C}^{*} = \min_{\mathbf{C}} \sum_{j=1}^{m} \sum_{i \in \mathcal{F}_j} p_i \lVert \mathbf{V}_j [i, :] - \mathbf{C}[j, :] \rVert_2^2.
\end{equation}
By applying Cauchy–Bunyakovsky–Schwarz inequality to the above equation, we receive the required Eq. \ref{eq:codebook-target}.

\section{Further Details of Shaver}\label{sec:shaver_details}

\subsection{Pseudo code}
\label{sec:pseudo-code}
The pseudo code is shown in Algorithm \ref{alg:sage-row}. For each data point, we first sample a random order $\mathcal{R}$ for set of players in $u(\cdot)$ (line \textbf{3}). We initialize an empty set to track the selected players (line \textbf{4}) and record the value generated by local value function $u(\cdot)$ for the empty set (line \textbf{5}). 
Then, for each player in the randomly ordered list, we add them to the selection set $\mathcal{S}$ (line \textbf{7}) and recalculate the value with this new team $\mathcal{S}$ (line \textbf{8}). With Theorem \ref{theorem:local-eq}, we update the corresponding player $(i,c)$'s contribution accordingly (lines \textbf{9-10}), and move on to the next player in the permutation. 
After iterating over all data points, the algorithm normalizes accumulated Shapley values with the total number of data samples in the dataset (line \textbf{14}).
This algorithm has the same magnitude of variance as a naive Monte Carlo approximation while reducing the sampling cost (Appendix \ref{sec:var-prove}). 
To validate the low variance of our approximation method, an algorithm convergence rate analysis is provided in Appendix \ref{sec:var-prove}.

\RestyleAlgo{ruled}
\LinesNumbered
\begin{algorithm}[t!]
\setstretch{0.8} 
\caption{Efficient Shapley Value Approximation for Embedding Parameters}
\label{alg:sage-row}
\SetKwRepeat{Do}{do}{while}
\SetKwComment{Comment}{/* }{ */}

\KwData{A trained CTR model $f(\cdot)$, dataset $\mathcal{D}$, local value function $u(\cdot)$, number of feature fields $m$, number of features $n$, embedding dimension $d$}
\KwResult{Shapley value of each player $\{{\phi}(i,c)\}_{i=1,c=1}^{n,d}$}
$\{{\phi}(i,c)\}_{i=1,c=1}^{n,d} \gets \mathbf{0}$\;
\For{$(\mathbf{x}, y) \in \mathcal{D}$}{ 
    $\mathcal{R} \gets$ a random order of players $\{(j,c)\}_{j=1,c=1}^{m,d}$\;
    $\mathcal{S} \gets \emptyset$\;
    $u_{prev} \gets u(\mathcal{S}| \mathbf{x}, y)$\;
    \For{$(j,c) \in \mathcal{R}$}{
    $\mathcal{S}\gets \mathcal{S}\cup (j,c)$\;
        $u_{next} \gets  u(\mathcal{S}|\mathbf{x}, y)$\;
        Find $(\!i,c)$ as $(\!j,c)$'s corresponding player w.r.t. $(\mathbf{x}, y)$\;
        $\phi(i,c) \gets \phi(i,c)  + u_{next} - u_{prev}$\;
        $u_{prev} \gets u_{next}$\;
    }
}
$\{{\phi}(i,c)\}_{i=1,c=1}^{n,d} \gets \{\frac{1}{|\mathcal{D}|}{\phi}(i,c)\}_{i=1,c=1}^{n,d}$
\end{algorithm}

\subsection{Algorithm Analysis}
\label{sec:var-prove}

This part provides prove for the converge rate of Algorithm. \ref{alg:sage-row}, based on \cite{sage2020}.
The central limit theorem states that if dataset $\mathcal{D}$ becomes large, the sample mean $\hat{\phi}_v(i)$ converges in distribution to a Gaussian with mean:
\begin{equation}
\label{eq:unbiased}
    \meanE \left[ 
\Delta(j,c|\mathcal{S}_u)
\right] = \meanE \left[ \Delta(i,c | \mathcal{S}_v) \right] = \phi_v(i,c).
\end{equation}
The first equation is due to Theorem \ref{theorem:local-eq}, while the second equation is due to Eq. (\ref{eq:sage-eq}). 
Eq. (\ref{eq:unbiased}) shows that Algorithm \ref{alg:sage-row} is unbiased.
Also because of central limit theorem, the estimation variance is:

\begin{equation*}
    \frac{\text{Var} \left( \Delta (j,c|\mathcal{S}_u) \right)}{\len{\mathcal{D}}}.
\end{equation*}
Although we do not have access to the numerator $\text{Var} \left( \Delta (j,c|\mathcal{S}_u) \right)$, we can conclude that the variance behaves as $O(1 / \len{\mathcal{D}})$. Thus, the algorithm would converge with large enough $\len{\mathcal{D}}$.

\subsection{Shapley Values on Synthetic Data}

To further demonstrate that our approximation aligns with the standard Shapley value, we tested on a 40-sample toy dataset with 3 feature fields, each of which has 2, 2, and 3 distinct features (7 features in total), respectively. We use the DeepFM model with a hidden size of 3, leading to $7 \times 3 = 21$
 embedding parameters.

Firstly, we compute the exact Shapley value of all 21 parameters, which are listed in table \ref{tab:true-shapley}. Then, we run our proposed method to estimate all Shapley values. Table \ref{tab:est-shapley} shows the estimated results. The absolute difference with the exact values is 0.0021 on average, verifying our unbiased estimation. Besides, our estimation takes substantially less time to run compared with exact Shapley values (20s vs. 3,849s).

\begin{table}[h]
\begin{tabular}{lllllll}
0.0180 & 0.0084 & 0.0010 & 0.0476 & 0.1787 & 0.0003 & 0.0312 \\
\midrule
0.0035 & 0.0167 & -0.0073 & 0.0162 & 0.0745 & 0.0038 & 0.2586 \\
\midrule
0.0228 & 0.0040 & 0.0122 & 0.0848 & 0.0081 & 0.0400 & -0.0001
\end{tabular}
\caption{True Shapley Value}
\label{tab:true-shapley}

\begin{tabular}{lllllll}
0.0179 & 0.0085 & 0.0062 & 0.0438 & 0.1756 & 0.0003 & 0.0292 \\
\midrule
0.0041 & 0.0203 & -0.0074 & 0.0150 & 0.0710 & 0.0090 & 0.2519 \\
\midrule
0.0253 & 0.0065 & 0.0136 & 0.0861 & 0.0079 & 0.0385 & -0.0003
\end{tabular}
\caption{Estimated Shapley value}
\label{tab:est-shapley}
\end{table}

\section{Experiment Settings}

In this section, we provide more details on our experiment settings.

\subsection{Baselines}
\label{sec:baselines}
We compare Shaver with the following embedding pruning methods for on-device content-based recommendation:
\begin{itemize}
    \item QR \cite{qr2019} divides the original embedding table into two smaller embedding tables (commonly called meta-embedding tables). The meta-embedding vectors are aggregrated by multiplication to create the final embedding vector. 
    \item TTRec \cite{ttrec2021} tackles the challenge of excessive embedding table parameters with a sequence of matrix products, employing tensor-train decomposition. 
    \item PEP \cite{pep2021} adopts a Soft-Thresholding Reparametrization (STR) trick to iteratively remove redundant parameters in the first training step.
    \item OptEmbed \cite{optembed2022} involves three consecutive steps: Train supernet and remove excessive features, evolutionary search for optimal embedding size for each field, and finally, retrain the network from the found embedding mask.
    \item CERP \cite{cerp2023} incorporate STR into two equal-size embedding tables $\mathbf{E}_1$ and $\mathbf{E}_2$. To compensate for the two sparse meta-embedding tables, the authors integrate a regularization loss into the training loss and combine two meta-embedding vectors with a sum operation, resulting in a dense vector. 
    \item MagPrune \cite{tran2024thorough} prune the trained embedding table by their magnitude. In line with our method, MagPrune also compresses the model by pruning the trained model.
    \item Taylor \cite{taylor2017} proposes to prune the model with the first-order Taylor approximation of the loss function. Similar to our method, this method also employs feedback from the dataset to calculate the attribution scores. 
    \item Post-Training Quantization (PTQ) \cite{ptq2016,ptq2021} first linearly scales the weight into $[-2^{b - 1}, 2^{b - 1} - 1]$ range, where $b$ is the bit width used. 
    Then we round up the result, converting parameters from float to int. 
    When performing inference, these compressed parameters would be scaled back to float32 with the stored bias and scale.
\end{itemize}

For Taylor and our methods, we pruned the original checkpoints as specified. For PTQ, we compressed the models from float32 to (4, 8, 16)-bit integers, the sparsity rate is calculated as $1 - b / 32$. During training, Shaver's codebook $\mathbf{C}$ is freezed. 

\subsection{Datasets}
\label{sec:dataset-preprocess}

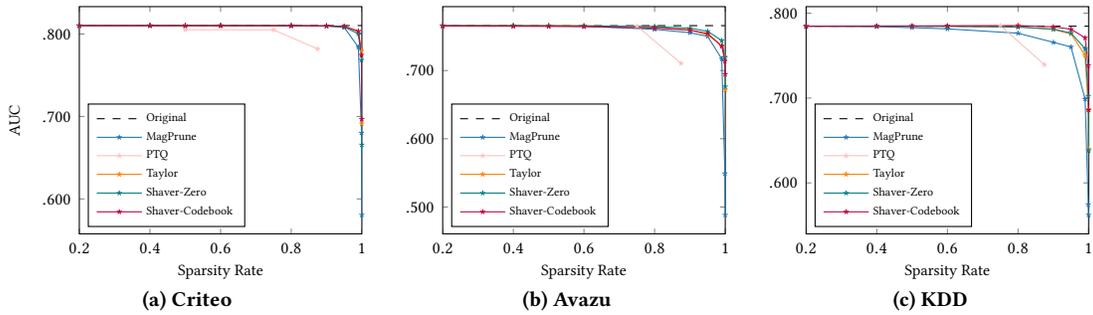
\begin{figure*}[thp]
    \centering
    \subcaptionbox{Criteo}
    {
        \vspace{-.2cm}
        \pgfplotstableread[col sep=comma]{figures/rq2/data/criteo_deepfm.csv}\criteodcn
\definecolor{blue1}{rgb}{0.121569,0.466667,0.705882}%
\begin{tikzpicture}
\tikzstyle{every node}=[font=\scriptsize]

\begin{axis}[
    xlabel={Sparsity Rate},
    ylabel={AUC},
    width=0.3\textwidth,
    xmin=0.2, xmax=1.0,
    xtick={0.2, 0.4, 0.6, 0.8, 1.0},
    y label style={at={(axis description cs:-0.18,.5)}},
    x label style={at={(axis description cs:0.5,-0.1)}},
    yticklabel style={
        /pgf/number format/precision=3,
        /pgf/number format/fixed zerofill,
        /pgf/number format/skip 0.=true
    },
    every x tick label/.append style={font=\scriptsize},
    legend pos=south west,
    legend style={nodes={scale=0.75}},
    ymajorgrids=false,
    grid style=dashed,
    tickwidth=0.07cm,
    every axis plot/.append style={
        mark=star,
        mark options={scale=0.5},
    },
    legend cell align=left,
]
\legend{Original, MagPrune, PTQ, Taylor, Shaver-Zero, Shaver-Codebook}

\addplot[
    color=black,
    dashed,
    mark=none,
] coordinates {
   (0.0,0.8102) (1.0,  0.8102)
};

\addplot[
    color=blue1,
] table[
    x=sparse-rate, 
    y=mag
]
{\criteodcn};

\addplot[
    color=pink,
] table[
    x=sparse-rate, 
    y=ptq
]
{\criteodcn};

\addplot[
    color=orange,
] table[
    x=sparse-rate, 
    y=taylor.pth
]
{\criteodcn};

\addplot[
    color=teal,
] table[
    x=sparse-rate, 
    y=zero.pth
]
{\criteodcn};

\addplot[
    color=purple,
] table[
    x=sparse-rate, 
    y=codebook.pth
]
{\criteodcn};

\end{axis}
\end{tikzpicture}
    }
    \subcaptionbox{Avazu}
    {
        \vspace{-.2cm}
        \pgfplotstableread[col sep=comma]{figures/rq2/data/avazu_deepfm.csv}\avazudcn
\definecolor{blue1}{rgb}{0.121569,0.466667,0.705882}%
\begin{tikzpicture}
\tikzstyle{every node}=[font=\scriptsize]

\begin{axis}[
    xlabel={Sparsity Rate},
    width=0.3\textwidth,
    xmin=0.2, xmax=1.0,
    xtick={0.2, 0.4, 0.6, 0.8, 1.0},
    y label style={at={(axis description cs:-0.18,.5)}},
    x label style={at={(axis description cs:0.5,-0.1)}},
    yticklabel style={
        /pgf/number format/precision=3,
        /pgf/number format/fixed zerofill,
        /pgf/number format/skip 0.=true
    },
    every x tick label/.append style={font=\scriptsize},
    legend pos=south west,
    legend style={nodes={scale=0.75}},
    ymajorgrids=false,
    grid style=dashed,
    tickwidth=0.07cm,
    every axis plot/.append style={
        mark=star,
        mark options={scale=0.5},
    },
    legend cell align=left,
]
\legend{Original, MagPrune, PTQ, Taylor, Shaver-Zero, Shaver-Codebook}

\addplot[
    color=black,
    dashed,
    mark=none,
] coordinates {
   (0.0, 0.7658) (1.0,  0.7658)
};

\addplot[
    color=blue1,
] table[
    x=sparse-rate, 
    y=mag
]
{\avazudcn};

\addplot[
    color=pink,
] table[
    x=sparse-rate, 
    y=ptq
]
{\avazudcn};

\addplot[
    color=orange,
] table[
    x=sparse-rate, 
    y=taylor2.pth
]
{\avazudcn};

\addplot[
    color=teal,
] table[
    x=sparse-rate, 
    y=zero.pth
]
{\avazudcn};

\addplot[
    color=purple,
] table[
    x=sparse-rate, 
    y=codebook.pth
]
{\avazudcn};

\end{axis}
\end{tikzpicture}
    }
    \subcaptionbox{KDD}
    {
        \vspace{-.2cm}
        \pgfplotstableread[col sep=comma]{figures/rq2/data/kdd_deepfm.csv}\kdddcn
\definecolor{blue1}{rgb}{0.121569,0.466667,0.705882}%
\begin{tikzpicture}
\tikzstyle{every node}=[font=\scriptsize]

\begin{axis}[
    xlabel={Sparsity Rate},
    width=0.3\textwidth,
    xmin=0.2, xmax=1.0,
    xtick={0.2, 0.4, 0.6, 0.8, 1.0},
    y label style={at={(axis description cs:-0.18,.5)}},
    x label style={at={(axis description cs:0.5,-0.1)}},
    yticklabel style={
        /pgf/number format/precision=3,
        /pgf/number format/fixed zerofill,
        /pgf/number format/skip 0.=true
    },
    every x tick label/.append style={font=\scriptsize},
    legend pos=south west,
    legend style={nodes={scale=0.75}},
    ymajorgrids=false,
    grid style=dashed,
    tickwidth=0.07cm,
    every axis plot/.append style={
        mark=star,
        mark options={scale=0.5},
    },
    legend cell align=left,
]
\legend{Original, MagPrune, PTQ, Taylor, Shaver-Zero, Shaver-Codebook}

\addplot[
    color=black,
    dashed,
    mark=none,
] coordinates {
   (0.0, 0.7847) (1.0,  0.7847)
};

\addplot[
    color=blue1,
] table[
    x=sparse-rate, 
    y=mag
]
{\kdddcn};

\addplot[
    color=pink,
] table[
    x=sparse-rate, 
    y=ptq
]
{\kdddcn};

\addplot[
    color=orange,
] table[
    x=sparse-rate, 
    y=taylor.pth
]
{\kdddcn};

\addplot[
    color=teal,
] table[
    x=sparse-rate, 
    y=zero.pth
]
{\kdddcn};

\addplot[
    color=purple,
] table[
    x=sparse-rate, 
    y=codebook.pth
]
{\kdddcn};

\end{axis}
\end{tikzpicture}
    }
    \vspace{-.2cm}
    \caption{Comparative results with single-shot baselines for DeepFM backbone.}
    \label{fig:rq2deepfm}
\end{figure*}

The dataset preprocess steps are described below:
\begin{itemize}
    \item \textbf{Criteo} \cite{Criteo} contains ad click data over a week. 
    Following the winning solution from the original Criteo Challenge, we discretize each value $x$ in numeric feature fields to $\lceil {\text{log}_2(x)} \rceil$ if $x > 2$.
    We replace infrequent features, which appear less than min\_count = $10$ times, with out-of-vocabulary (OOV) tokens for each field.
    \item \textbf{Avazu} \cite{Avazu} includes 10 days of click-through data.
    As a common practice, we first remove the `id' field, which has a unique value for every record. 
    Similar to Criteo dataset, we remove infrequent features and replace them with OOV tokens (min\_count = 2).
    \item \textbf{KDD} \cite{KDD} consists of records gathered from search session logs. Similar to previous two datasets, we utilize OOV tokens to represent infrequent features (min\_count = 10).
\end{itemize}

For Criteo and Avazu dataset, we directly adopt the split and preprocessed data from \cite{tran2024thorough}.

\section{Further Experimental Results}
\label{sec:more-exp}

\subsection{DeepFM Backbone Experiment Results}


 \begin{table}[htb]
 \centering
 \caption{Original (Uncompressed) Model Performance}
 \vspace{-.3cm}
     \begin{tabular}{ccllr}
     \toprule
     Dataset & Backbone & AUC & Loss & Params\\
     \midrule
     \multirow{2}{*}{Criteo}
     & DeepFM & 0.8102 & 0.4416 & 17.39M \\
     & DCN & 0.8114 & 0.4407 & 17.39M \\
     \midrule
    \multirow{2}{*}{KDD}
     & DeepFM & 0.7847 & 0.1563 & 96.31M \\
     & DCN & 0.7864 & 0.1562 & 96.31M \\
     \midrule
     \multirow{2}{*}{Avazu}
     & DeepFM &0.7658 & 0.3932 & 70.86M \\
     & DCN & 0.7759 & 0.3839& 70.86M \\
     \bottomrule
     \end{tabular}
     \label{tab:original}
 \end{table}

\renewcommand{\arraystretch}{0.8}
\setlength\tabcolsep{3.5pt}
\begin{table*}[htbp]
    \centering
        \caption{Comparative results with non-single-shot baselines under three sparsity rates. All methods adopt DeepFM as the CRS backbone. In each dataset, the best result is marked in bold and the second best one is underlined.}
    \vspace{-.3cm}
\begin{tabular}{llcccccccccc}
    \toprule
    \multirow[c]{2}{*}{Dataset} & \multirow[c]{2}{*}{Method} &Single- & \multicolumn{3}{c}{$t$ = 50\%} & \multicolumn{3}{c}{$t$ = 80\%} & \multicolumn{3}{c}{$t$ = 95\%} \\
    \cmidrule(lr){4-6}\cmidrule(lr){7-9}\cmidrule(lr){10-12}
     & & shot? & AUC$\uparrow$ & LogLoss$\downarrow$ & \#Params & AUC$\uparrow$ & LogLoss$\downarrow$ & \#Params & AUC$\uparrow$ & LogLoss$\downarrow$ & \#Params \\
    \midrule

\multirow[c]{8}{*}{Criteo} 
     & QR \cite{qr2019} & \xmark & 0.8081 & 0.4435 & 8.69M & 0.8078 & 0.4438 & 3.48M & 0.8033 & 0.4482 & 870K \\
     & TTRec \cite{ttrec2021} & \xmark & 0.8075 & 0.4442 & 9.36M & 0.8070 & 0.4446 & 3.15M & 0.8087 & 0.4431 & 870K \\
     & PEP \cite{pep2021} & \xmark & \textbf{0.8105} & \textbf{0.4414} & 8.67M & 0.8098 & 0.4419 & 3.47M & 0.8084 & 0.4432 & 806K \\
     & OptEmb \cite{optembed2022} & \xmark & \multicolumn{1}{c}{-} & \multicolumn{1}{c}{-} & \multicolumn{1}{c}{-} & 0.8088 & 0.4430 & 3.17M & 0.8043 & 0.4471 & 863K \\
     & CERP \cite{cerp2023} & \xmark & 0.8099 & 0.4419 & 8.66M & 0.8095 & 0.4423 & 3.21M & 0.8062 & 0.4454 & 864K \\
     \cmidrule(lr){2-12}
     & Shaver-Zero & \cmark & 0.8101 & 0.4417 & 8.69M & 0.8100 & \underline{0.4418} & 3.48M & 0.8085 & 0.4432 & 869K \\
     & Shaver-Zero$\dagger$ & \xmark & \underline{0.8103} & \underline{0.4415} & 8.69M & \textbf{0.8105} & \textbf{0.4414} & 3.48M & \underline{0.8104} & \underline{0.4414} & 869K \\
      & Shaver-Codebook & \cmark & 0.8101 & 0.4417 & 8.69M & \underline{0.8101} & \underline{0.4418} & 3.48M & 0.8093 & 0.4425 & 870K \\
     & Shaver-Codebook$\dagger$ & \xmark & \textbf{0.8105} & \textbf{0.4414} & 8.69M & \textbf{0.8105} & \textbf{0.4414} & 3.48M & \textbf{0.8106} & \textbf{0.4413} & 870K \\

\midrule

\multirow[c]{8}{*}{Avazu} 
   & QR \cite{qr2019} & \xmark & \underline{0.7697} & \underline{0.3864} & 35.43M & \underline{0.7695} & \underline{0.3871} & 14.17M & \underline{0.7644} & \underline{0.3905} & 3.54M \\
  & TTRec \cite{ttrec2021} & \xmark & 0.7647 & 0.3915 & 35.43M & 0.7651 & 0.3902 & 14.17M & 0.7608 & 0.3924 & 3.54M \\
   & PEP \cite{pep2021} & \xmark & 0.7647 & 0.3886 & 29.67M & 0.7633 & 0.3896 & 14.07M & 0.7590 & 0.3922 & 3.50M \\
   & OptEmb \cite{optembed2022} & \xmark & \multicolumn{1}{c}{-} & \multicolumn{1}{c}{-} & \multicolumn{1}{c}{-} & \multicolumn{1}{c}{-} & \multicolumn{1}{c}{-} & \multicolumn{1}{c}{-} & 0.7585 & 0.3936 & 3.41M \\
& CERP \cite{cerp2023} & \xmark & 0.7649 & 0.3885 & 35.30M & 0.7638 & 0.3892 & 14.14M & 0.7607 & 0.3920 & 3.54M \\
\cmidrule(lr){2-12}
 & Shaver-Zero & \cmark & 0.7655 & 0.3934 & 35.43M & 0.7645 & 0.3942 & 14.17M & 0.7575 & 0.3984 & 3.54M \\
 & Shaver-Zero$\dagger$ & \xmark & \textbf{0.7715} & \textbf{0.3855} & 35.43M & \textbf{0.7715} & \textbf{0.3857} & 14.17M & \textbf{0.7655} & \textbf{0.3902} & 3.54M \\
  & Shaver-Codebook & \cmark & 0.7647 & 0.3939 & 35.43M & 0.7623 & 0.3952 & 14.17M & 0.7543 & 0.4002 & 3.54M \\
 & Shaver-Codebook$\dagger$ & \xmark & 0.7644 & 0.3903 & 35.43M & 0.7602 & 0.3940 & 14.17M & 0.7552 & 0.3974 & 3.54M \\

\midrule

\multirow[c]{8}{*}{KDD} 
& QR \cite{qr2019} & \xmark & 0.7744 & 0.1579 & 48.16M & 0.7753 & 0.1578 & 19.26M & 0.7779 & 0.1573 & 4.82M \\
  & TTRec \cite{ttrec2021} & \xmark & 0.7762 & 0.1588 & 48.16M & 0.7727 & 0.1583 & 19.26M & \underline{0.7785} & \underline{0.1572} & 4.82M \\
   & PEP \cite{pep2021} & \xmark & 0.7801 & 0.1568 & 47.80M & 0.7745 & 0.1578 & 19.22M & 0.7730 & 0.1579 & 4.71M  \\
   & OptEmb \cite{optembed2022} & \xmark & 0.7742 & 0.1584 & 47.52M & 0.7740 & 0.1584 & 18.29M
   & 0.7663 & 0.1597 & 4.79M  
   \\
& CERP \cite{cerp2023} & \xmark & 0.7818 & 0.1567 & 48.10M & 0.7765 & 0.1578 & 18.57M & 0.7727 & 0.1581 & 4.77M \\
\cmidrule(lr){2-12}
 & Shaver-Zero & \cmark & \underline{0.7846} & \underline{0.1563} & 48.16M  & \underline{0.7836} & \underline{0.1565} & 19.26M & 0.7766 & 0.1578 &  4.82M\\
 & Shaver-Zero$\dagger$ & \xmark & \underline{0.7846} & \underline{0.1563} & 48.16M & \underline{0.7836} & \underline{0.1565} & 19.26M & 0.7766 & 0.1578 & 4.82M\\
  & Shaver-Codebook & \cmark & \textbf{0.7850} & \textbf{0.1562} & 48.16M & \textbf{0.7854} & \textbf{0.1560} & 19.26M & \textbf{0.7812} & \textbf{0.1567} & 4.82M\\
 & Shaver-Codebook$\dagger$ & \xmark & \textbf{0.7850} & \textbf{0.1562} & 48.16M & \textbf{0.7854} & \textbf{0.1560} & 19.26M & \textbf{0.7812} & \textbf{0.1567} & 4.82M \\
 
\bottomrule
\end{tabular}
    \vspace{-.3cm}

\label{tab:RQ1deepfm}
\end{table*}
\renewcommand{\arraystretch}{1.0}

Table \ref{tab:original} provides uncompressed model performance. Table \ref{tab:RQ1deepfm} and Figure \ref{fig:rq2deepfm} show the comparison results with other non-single-shot and single-shot methods of Shaver for DeepFM backbone.

For DeepFM backbone, our methods achieves the best results in all settings.
With the DeepFM backbone, where the original model has poor performance, our method can remove excessive parameters and improve model performance with fine-tuning. Thus, model fine-tuning has a more positive effect on the DeepFM backbone.
Except for the Avazu dataset, codebook approaches outperform zero baseline slightly on Criteo, and significantly on KDD. The reason for this exception is unclear. 
One plausible explanation is that zero padding can act as a regularization and better prevent overfitting compared to codebook approaches. 

Regarding single-shot comparison, in general, the performance trend is similar to DCN, with few exceptions. First, PTQ achieves worse relative results in the Criteo dataset. Second, as mentioned above, Shaver-Zero outperforms Shaver-Codebook and other methods in Avazu settings.

\subsection{Model Efficiency on Edge Devices}

To further demonstrate the real-world applicability of our method on edge devices, we deploy Shaver into three devices: (1) Raspberry Pi 4B, (2) a personal PC with CPU i7-13700K, (3) a Samsung Galaxy S23 with Snapdragon 8 Gen 2 CPU.

\begin{table}[t]
\caption{Inference efficiency benchmark results. ``Mem.'' represents inference memory peak. Time is shown in seconds, while storage and memory is shown in MiB.}
\label{tab:efficiency-cpu}
\begin{tabular}{crrrrrrr}
\toprule
\multirow[c]{2}{*}{t} & \multirow[c]{2}{*}{Storage} & \multicolumn{2}{c}{Raspberry} & \multicolumn{2}{c}{CPU} & \multicolumn{2}{c}{Android} \\
 \cmidrule(lr){3-4}  \cmidrule(lr){5-6}  \cmidrule(lr){7-8}
& & Mem. & Time & Mem. & Time & Mem. & Time \\
\midrule
80\% & 165 & 393.91 & 0.0211 & 343.125 & 0.0166 & 312.2 & 0.0419 \\
90\% & 98 & 246.953 & 0.0210 & 212.922 & 0.0167 & 191.2 & 0.0421 \\
95\% & 64 & 153.992 & 0.0205 & 155.461 & 0.0168 & 133.4 & 0.0446 \\
99\% & 38 & 103.972 & 0.0204 & 133.195 & 0.0166 & 86.2 & 0.0393 \\
\bottomrule
\end{tabular}
\end{table}

We implement our algorithm using Python with PyTorch on Raspberry Pi and PC, and Java with ExecuTorch on mobile devices.
To verify that the efficiency is able to support on-device deployment, on each device, we report the memory inference memory and inference time of the compressed DCN-Mix CRS backbone model. The test is performed with the largest dataset KDD. The inference memory is measured with memory profiler (for Python: https://pypi.org/project/memory-profiler/, for Java: Android Studio built-in memory profiler). Regarding the inference time, we take the average of 10 runs with a batch size of 64. We further report the model storage requirement as it is vital for transferring bandwidth and devices with limited disk storage \cite{han2021deeprec}. The sparse tensor structure is CSR.

The on-device results are shown table \ref{tab:efficiency-cpu}. Firstly, it is worth noting that the inference time does not change dramatically between sparsity settings, it is mainly determined by the backbone CRS model’s complexity. Secondly, we can observe a linear relationship between the model inference memory and the sparsity rate. Lastly, an overall conclusion is that CRS models pruned by Shaver is able to achieve a highly practical memory and inference time for when being deployed for on-device services.

\end{document}